 \documentclass[]{imsart}

\RequirePackage{amsthm,amsmath,amsfonts,amssymb}
\RequirePackage[authoryear]{natbib}
\RequirePackage[colorlinks,citecolor=blue,urlcolor=blue]{hyperref}
\RequirePackage{graphicx}

\usepackage{enumitem}
\usepackage{nameref}

\makeatletter 
\def\namedlabel#1#2{\begingroup
    #2%
    \def\@currentlabel{#2}%
    \phantomsection\label{#1}\endgroup
}

\startlocaldefs
\theoremstyle{plain}

\newtheorem{theorem}{Theorem}[section]
\theoremstyle{remark}

\newtheorem{remark}{Remark}

\endlocaldefs

\begin{document}

\begin{frontmatter}
\title{Review of Quasi-Randomization Approaches for Estimation from Non-probability Samples}
\runtitle{Review of Quasi-Randomization Approaches}

\begin{aug}
\author[A]{\fnms{Vladislav}~\snm{Beresovsky} 
\ead[label=e1]{beresovsky.vladislav@bls.gov}},
\author[B]{\fnms{Julie}~\snm{Gershunskaya}\ead[label=e2]{gershunskaya.julie@bls.gov}}
\and
\author[C]{\fnms{Terrance D.}~\snm{Savitsky}\ead[label=e3]{savitsky.terrance@bls.gov} 
}
\runauthor{Beresovsky et al.}


\address[A]{Vladislav Beresovsky is Research Mathematical Statistician,
U.S. Bureau of Labor Statistics\printead[presep={\ }]{e1}.}

\address[B]{Julie Gershunskaya is Mathematical Statistician, 
U.S. Bureau of Labor Statistics\printead[presep={\ }]{e2}.}

\address[C]{Terrance D. Savitsky is Senior Research Mathematical Statistician, U.S. Bureau of Labor Statistics\printead[presep={\ }]{e3}.
}
\end{aug}

\begin{abstract}
The recent proliferation of computers and the internet have opened new opportunities for collecting and processing data. However, such data are often obtained without a well-planned probability survey design. Such non-probability based samples cannot be automatically regarded as representative of the population of interest. Several classes of  
methods for estimation and inferences from non-probability samples have been developed in recent years. 
The quasi-randomization methods assume that non-probability sample selection is governed by an underlying latent random mechanism. The basic idea is to use information collected from a probability (“reference”) sample to uncover latent non-probability survey participation probabilities (also known as “propensity scores”) and use them in estimation of target finite population parameters. In this paper, we review and compare theoretical properties of recently developed methods of estimation survey participation probabilities and study their relative performances in simulations.
\end{abstract}

\begin{keyword}
\kwd{Data combining}
\kwd{non-probability sample}
\kwd{reference sample}
\kwd{participation probabilities}
\kwd{sample likelihood}
\kwd{variance estimation}
\end{keyword}
\end{frontmatter}

\section{Introduction} \label{sec:intro}



Over the late 19th and beginning of 20th century, survey statisticians discussed the appropriateness of using partial enumeration, or samples, for studying populations. The debates aroused after Anders Nicolai Kiær, the founder and Director of the Norwegian Bureau of Statistics, presented his method of sampling at the 1895 meeting of the International Statistical Institute. Kiær called the method ``representative'', since units were carefully (non-randomly) selected from the population, based on a number of important demographic characteristics. His peers debated whether such samples indeed represented the population of interest and what criteria would be appropriate to gauge the representativeness of a sample (\citet{Kruskal1980}). The debates settled down in 1934, when Jerzy Neyman published a seminal paper in which he criticised the purposive, non-probability based, way of sampling and introduced his sampling philosophy with the central idea of the randomization based sample selection process as the essential prerequisite to ensure representativeness of a sample, see \citet{neyman1934}. Since then, \emph{randomization} became a widely accepted way for obtaining representative samples to estimate target population parameters.


Recent years have seen a renewed and growing interest in \emph{non-randomized} samples to augment survey estimation due to an ongoing reduction in survey response rates and the rise of technology that connects people on-line for building data repositories. 
Although it is understood that only a randomization based approach provides guarantees for the representativeness of a sample, traditional well-planned probability based surveys may not be easy to obtain. A relatively expensive and time-consuming data collection process for probability surveys is highly labor intensive, even more so as agencies attempt to partially remedy non-response.
There has been a long-term trend towards reductions in response rates for nearly all randomized survey instruments administered by government statistical agencies \citep{10.1093/jssam/smx019}. 
In contrast, the recent proliferation of computers and the internet has created new ways for acquiring data.  A lot of information now can be collected seamlessly and is often available as a simple by-product of other business activities. Such data sometimes are called ``convenience'' samples. The demand for exploiting these resources is steadily growing. The downside of such ``opportunistic'' collection of information is that, no matter how big the samples are, data obtained without a well-planned probability survey design cannot be automatically regarded as being representative of the population of interest. \citet{2010Bethlehem}, \citet{2018Meng}, \citet{2011VanderWeele_Shpitser}, among others, warn that naïve estimation from non-probability based sources may lead to substantial biases.
The current interest to use non-randomized data is not to replace the randomized survey samples, but rather to leverage the randomized survey samples to remove potential estimation bias from these, typically large-sized, non-randomized data sources and improve estimation quality (through reducing estimator variances).

The fundamental ideas for making inferences from non-probability samples are parallel to those developed within observational studies and missing data literature (\citet{Rubin1976}, \citet{RosenbaumRubin1983}) and are based on modeling assumptions that seek to account for the sample participation mechanism to reduce estimation bias.  Existing methodology for estimation from non-probability samples 
can be broadly grouped into model-based, quasi-randomization, and doubly-robust approaches. Model-based approaches are based on a model for the target variable of interest, whereas quasi-randomization approaches focus on estimation of the propensity to participate in the non-probability survey; doubly-robust estimators combine these two types of models.   \citet{2017elliot}, \citet{2020Valliant},  \citet{2021Beaumont_Rao}, \citet{2020Rao}, \citet{2022Wu} provide a comprehensive account of this rapidly developing field. 

In the current paper, we focus on a set of quasi-randomization approaches. 
These methods assume that non-probability sample selection is governed by an underlying latent random mechanism. The basic idea is to use information collected from a probability based sample, drawn independently from the same target population, to uncover latent non-probability survey participation probabilities and use them in the estimation of target finite population parameters. In this context, the required probability based sample is often regarded as ``reference'' sample. 

The important requirement is that both the non-probability and reference samples contain a common set of covariates and the assumption is that the non-probability sample participation probabilities are governed by these covariates and could be estimated from a model. \citet{2020Valliant} also notes the importance of  the \emph{common support} requirement, which means that the target population units are assumed to have positive probabilities to be included into both samples conditional on the shared set of covariates. For a broad range of typically used covariates, such as large demographic groups, there is usually enough overlap in the characteristics of units, so that the common support requirement is satisfied for such covariates. It becomes less tenable, however, for more detailed classifications and for covariates identifying smaller population groups. In this case, we are faced with a potentially insufficient representation of some population groups in either of the two samples. Thus, it is important to develop an efficient and robust methodology that would lead to better utilization of available detailed covariates for more accurate estimation of participation probabilities. 

We review several methods for estimation of survey participation probabilities and discuss our likelihood-based approach that is more efficient than pseudo-likelihood based approaches for practically relevant scenarios. This paper can be considered as a sequel to \citet{savitsky2023} who used a Bayesian framework for the modeling. In the current paper, we derive asymptotic results and compare theoretical properties of the discussed methods, and explore their relative performances using a numerical study and simulations. 

{\citet{2020_ChenLiWu} proposed to directly model the unknown non-probability sample participation probabilities in a likelihood defined on the target population. Since the target population is not observed, they employ a pseudo-likelihood approach using the probability reference sampling weights. The role of the reference sample is limited to estimation of the log-likelihood term containing the  sum over the target population. 
The method has a benefit of being optimal when the reference sample coincides with the target population. However, \citet{2021valliant} noted that this method can be relatively inefficient when the reference sample is small and the  non-probability sample is large. This scenario is practically important when making estimates from large administrative data. 
}

{Another group of methods consider combined non-probability and probability samples and define indicator variables $I_z$ that take the value of 1 for  the non-probability sample units, and 0 for the reference sample units. Since indicator $I_z$ is observed, it is straightforward to model probability $\pi_z = P\{I_z=1\}$ on the combined samples set. Then, non-probability sample participation probabilities are obtained from a relationship between $\pi_z$ and the respective samples inclusion probabilities. 
}
{ \citet{2009elliot} and \citet{2017elliot} derived this relationship under the assumption that the samples inclusion probabilities are small, so that the  overlap between probability and non-probability samples is negligible. 
}

{
\citet{2021valliant} proposed the Adjusted Logistic Propensity (ALP) weighting method that is based on an imaginary construct stacking the non-probability sample and the target population. They formulated the Bernoulli likelihood on this set for indicator variable $I_z$ taking the value of 1 on
the non-probability sample part of the stack, and 0 on the target population. 
Probability $\pi_z=P\{I_z=1\}$ can be estimated from  
the respective pseudo-likelihood by using weighted logistic regression implemented in standard statistical software.
Then, the response propensity to the non-probability survey can be inferred from a formula linking the response propensity to the probability of being in the non-probability sample part of the stack.
}

{
The cited stack-based methods estimate response propensity in two steps. First, the probability $\pi_z$ is estimated on the stacked structure. Second, propensity to participate in the non-probability sample is derived from the relationship between $\pi_z$ and the respective samples inclusion probabilities. Such a two-step procedure does not guarantee optimal estimation of participation probabilities and does not necessarily produce estimates below 1. 
}

{\citet{savitsky2023} propose estimating participation probabilities from the stacked non-probability and reference samples setup. 
For completeness, we repeat the description of the method and the proof of the key formula in the current paper.
Similar to \citet{2009elliot} and \citet{2017elliot}, indicator variables $I_z$ are defined so that they take the value of 1 for units coming from the non-probability sample part of the stack, and 0 for units coming from the reference sample. Yet we note that if there are overlapping units, they would be included in the stack \emph{twice}, with indicator values of 1 and 0. 
The key difference between the two methods is that the new derivation of the relationship between $\pi_z$ and sample inclusion probabilities does not require the assumption of a negligible overlap between the two samples: under our stacked samples setup, the relationship holds \emph{exactly} for any degree of the overlap. For easy reference, we name the formula the Core Relationship for Independent Sampling Probabilities (CRISP), due to the central role it plays in the estimation of the participation probabilities.

There is a common perception in the literature that under the stacked samples approach the overlapping units must be removed (for instance, in the case of making estimates from a large administrative dataset treated as a non-probability sample).  The problem is that, in practice, we often do not know which units belong to both samples. \citet{2022LiuScholtusWaal} proposed to identify the common units via probabilistic record linkage under the assumption that it can be done with high probability, and \citet{2021KimTam} used a nonparametric classification method.
We provide the proof that, under usual assumptions, indicator variables $I_z$ on the stacked set are practically independent conditional on covariates, even when the probability of the overlap is large. This allows us to use the likelihood for independent Bernoulli variables formulated over all units in the stacked sample and avoid the need to identify and remove the overlapping units.

 Another important innovation is in that we use CRISP to parameterize the probabilities in the logistic regression model by treating them as a composite function, so that the estimates of the non-probability sample participation probabilities could be obtained in one step directly from the model fit. \citet{beresovsky2019} used this way of estimation, called Implicit Logistic Regression (ILR). The ILR approach is an improvement over the two-step procedure of \citet{2009elliot} and \citet{2017elliot} as it provides better efficiency and avoids cases of estimated propensities exceeding 1. Analogously, we present a version of the ILR method for the setup and pseudo-likelihood approach of \citet{2021valliant}. This method, called pseudo-ILR (PILR), is an alternative to the ALP that estimates participation probabilities in a one-step procedure.   
}

In Section \ref{sec:setup}, we introduce the setup and notation for the modeling of convenience sample participation probabilities. In Section~\ref{sec:CLW_ALP}, we review the method of \citet{2020_ChenLiWu} that directly models
the unknown convenience sample inclusion probabilities, $\pi_{c}$, in a likelihood defined on the population.  Section~\ref{sec:basicformula} introduces an alternative class of methods that \emph{indirectly} models
$\pi_{c}$ in their likelihood statements through composing a functional relationship between probabilities of inclusion
from the population to probabilities of being in a given part of
the observed ``stacked'' convenience and reference samples. We state the CRISP formula 
and provide its proof that utilizes conditional independence of inclusions in the convenience and reference samples. 
In Section \ref{sec:ILR_PILR}, we describe methods for modeling convenience sample participation probabilities using this relationship.
Theoretical properties of the estimators of model parameters and respective inverse propensity weighted (IPW) estimators of the population mean are given in Section \ref{sec:theorems} for the $3$ methods on which we focus.  Section \ref{sec:sims} presents a simulation study that compares methods under various scenarios that include varying reference and convenience sample sizes and sampling fractions, as well as different degrees of the common support. We provide concluding remarks in Section \ref{sec:conclusion}.


\section{Estimation of convenience sample participation probabilities} \label{sec:models}
\subsection{Setup and notation} \label{sec:setup}

{In this paper we focus on estimation of the finite population mean ${\mu} ={N^{-1}} \mathop{\sum}_{i \in U}y_i$, for units $i = 1,\ldots,N$ in target population $U$, in the situation where variables of interest $y_i$ are observed only on the non-probability (convenience) sample $S_c$ of size $n_c$. In the sequel we often use the term ``convenience'' sample to generally refer to a non-probability sample. We assume that the observed convenience sample is a realization of a latent (unknown) random sampling design and define \emph{unknown} probability $\pi_{ci}=\pi_{c}(\mathbf{x}_i)=P\{I_{ci}=1|i \in U, \mathbf{x}_i\}$ for units in population $U$ to be included into the convenience sample.  Here, $I_{ci}$ is the inclusion indicator of population unit $i$, taking on the value of 1 if unit $i$ is included into $S_c$, and $0$ if it is in $U\backslash S_{c}$; $\mathbf{x}_i$ is a set of known covariates.
We refer to $\pi_{ci}$ as a ``participation'' probability (as differentiated from an inclusion probability) to recognize a typical realization of a convenience sample by self-selection, rather than solicitation.}

A probability (reference) sample $S_r$ of size $n_r$ is also assumed to be available. It is selected from the same finite population $U$ under a known probability survey design with inclusion probabilities $\pi_{ri}=P\{I_{ri}=1 \mid i \in U, \mathbf{x}^{'}_i\}$, where $I_{ri}$ is the reference sample inclusion indicator of unit $i$. 
Suppose samples $S_c$ and $S_r$ share a common set of covariates $\mathbf{x}_i = \mathbf{x}^{'}_i$. 
\\

{Similar to \citet{2020_ChenLiWu}, we assume the following about the selection process:}
\begin{enumerate}
\item [\namedlabel{itm:commsupp}{A1}:] 
 Each unit in the population has a positive participation probability: 
 $\pi_{ci}>0$ for all $i \in U.$ 
 \item [{A2}:] Sample $S_c$ selection mechanism is ignorable given $\mathbf{x}_i$:  $P\{I_{ci}=1| i \in U, \mathbf{x}_i, y_i\}=P\{I_{ci}=1|i \in U, \mathbf{x}_i\}$  for all $i \in U$. 
 \item[\namedlabel{itm:indepI_c} {A3}:] Indicators $I_{ci}$ and $I_{cj}$ are independent given $\mathbf{x}_i$ and $\mathbf{x}_j$ for $i \neq j.$
 \item[\namedlabel{itm:indepI_cI_r} {A4}:] Indicators $I_{ci}$ and $I_{ri}$ are independent given $\mathbf{x}_i$ for any $i \in U.$
\end{enumerate}

When the randomized sampling design and associated inclusion probabilties are known, 
a condition like Assumption ~\ref{itm:commsupp} ensures that the 
resulting design distribution (that governs the taking of repeated samples from the target population) will produce a collection of samples that together represent all population subgroups.  In our setting where the $\pi_{c}$ are estimated, Assumption~\ref{itm:commsupp} may be interpreted as a ``common support'' assumption that requires a complete overlap of the support of predictor values (or the distributions of predictor values) between the convenience and reference samples.   
See also discussion on a common support requirement in \citet{2020Valliant}.

Quasi-randomization approaches propose methods for estimation of non-probability sample participation probabilities $\pi_{ci}$ by combining non-probability and probability samples. These estimated probabilities are then used in the IPW estimator of the finite population mean.

\subsection{Direct estimation of participation probabilities \label{sec:CLW_ALP} }

We review the method of \citet{2020_ChenLiWu} for combining information from both samples to estimate the non-probability / convenience sample participation probabilities. Their method is in the class of pseudo-likelihood approaches where a log-likelihood is formulated over the population units, as if the whole finite population was observed. Next, since only a probability sample is observed, the log-likelihood is replaced by a pseudo log-likelihood formulation, using the probability sample weights from the observed reference sample. Model parameters are obtained as a solution to respective pseudo-likelihood based estimating equations. Finally, the estimated non-probability sample participation probabilities are used to form the IPW estimator of the finite population mean.    

\citet{2020_ChenLiWu} (hereafter CLW) write a log-likelihood over units in $U$, with respect to Bernoulli variable $I_{ci}$:      
\begin{eqnarray}
   && \ell^{CLW}(\boldsymbol{\beta}) = \label{eq:chen1}    \\ 
    &  & \mathop{\sum}_{i \in U}\left\{I_{ci}\log\left[{\pi_{ci}(\boldsymbol{\beta})}\right] +  (1-I_{ci})\log\left[1-\pi_{ci}\left(\boldsymbol{\beta}\right)\right]\right\}=  \nonumber \\ 
    &  & \mathop{\sum}_{i \in S_{c}}\log\left[\frac{\pi_{ci}(\boldsymbol{\beta})}{1-\pi_{ci}\left(\boldsymbol{\beta}\right)}\right] + \mathop{\sum}_{i \in U}\log\left[1-\pi_{ci}\left(\boldsymbol{\beta}\right)\right] \label{eq:chen2},
\end{eqnarray}
where $\boldsymbol{\beta}$ is the parameter vector in a logistic regression model  $\text{logit}\left[\pi_{ci}(\boldsymbol{\beta})\right]=\boldsymbol{\beta^T}\mathbf{x}_i$.  

We refer to this method as ``direct'' estimation because $\pi_{ci}$ is directly used in the Bernoulli likelihood  \eqref{eq:chen1} formulated for the \emph{unobserved} non-probabilty sample indicators $I_{ci}$.

Since population $U$ is not available, CLW employ a pseudo-likelihood approach by, first, regrouping units, as in (\ref{eq:chen2}), and then replacing the sum over the finite population with its probability sample based estimate:
\begin{eqnarray}\label{eq:chen3}
    && \hat{\ell}^{CLW}(\boldsymbol{\beta}) =\\
    && \mathop{\sum}_{i \in S_{c}}{\log\left[\frac{\pi_{ci}(\boldsymbol{\beta})}{1-\pi_{ci}\left(\boldsymbol{\beta}\right)}\right]} +    \mathop{\sum}_{i \in S_r}{w_{ri}\log\left[1-\pi_{ci}\left(\boldsymbol{\beta}\right)\right]},  \nonumber
\end{eqnarray}
where weights $w_{ri}=\pi_{ri}^{-1}$ are inverse values  of the reference sample inclusion probabilities $\pi_{ri}$. Estimates are obtained by solving respective pseudo-likelihood based estimating equations. 

Indicator $I_{ci}$ is typically not observed because the finite population is not generally available; in particular, one does not know which units from the finite population are selected into the convenience sample.  In the sequel we define an indicator for whether a unit from a pooled reference and convenience sample is included in the convenience sample, which is observed.  

We must tie the associated probability of membership in the pooled sample to the convenience sample participation probability that we seek to estimate using a model likelihood.  
In the next Section, we present such an expression.

\subsection{Indirect estimation of participation probabilities} \label{sec:basicformula}
We next derive a formula relating the probability of membership in the pooled sample to the convenience sample participation probability that we seek to estimate using a model likelihood.  This formula is used by the two methods introduced in Section \ref{sec:ILR_PILR} to estimate the latent convenience sample participation probabilities. The formula, named Core Relationship for Independent Sampling Probabilities (CRISP), is based on the assumption that, after conditioning on covariates, inclusions into the non-probability and reference samples are independent (see Assumption~\ref{itm:indepI_cI_r} in Section \ref{sec:setup}). 

Let us consider stacking together samples $S_c$ and $S_r$ and denote the resulting set by $S$. We define indicator variable $I_{zi}$ on such a combined sample $S$ as follows: $I_{zi}=1$ if unit $i$ belongs to the non-probability sample (part 1), and $I_{zi}=0$ if unit $i$ belongs to the probability sample (part 2). Under this construction, if there is an overlap between the two samples, $S_c$ and $S_r$, then the overlapping units are included into the stacked set, $S$, twice: as part of the non-probability sample (with $I_{zi}=1$) and as part of the reference probability sample (with $I_{zi}=0$). \citet{savitsky2023} use first principles to prove the relationship between probabilities $\pi_{zi}=P\{I_{zi}=1|i \in S, \mathbf{x}_i\}$ of being in part 1 of the stacked set, on the one hand, and inclusion probabilities, $\pi_{ci}$ and $\pi_{ri}$, on the other hand. For completeness, we reproduce this result in Theorem \ref{theo:2_1} of this Section. 
Furthermore, in Theorem \ref{theo:2_2} below, we show that, regardless of the overlap, under certain usual assumptions, indicators $I_{zi}$ and $I_{zj}$ are nearly conditionally independent for any two units $i$ and $j$ in  stacked set $S$.       

We note that the stacked samples setup should be contrasted with the combined sample approach of \citet{2009elliot} and \citet{2017elliot}, where they use the assumption that the overlap is negligible between the samples.  The stacked samples setup allows for an overlap of any size -- from no overlap to a situation where one sample is a subset of another; moreover, there is no need to know the size of the overlap or the identities of the units that belong to both samples. 

\begin{theorem}[CRISP] \label{theo:2_1}
Assuming samples $S_c$ and $S_r$ are independently selected from population U, the following relationship holds:
\begin{align} \label{eq:pi_z}
{\pi_{zi}}=\frac{{\pi_{ci}}}{{\pi_{ci}}+{\pi_{ri}}}.
\end{align}
\end{theorem}
\begin{proof}
Since the selection into respective samples is performed independently, the combined set $S$  can be viewed as emerging from the following imaginary stacked populations scheme displayed in Figure \ref{fig:scheme}. 
We consider two copies of target population $U$. We \emph{stack} the two copies of $U$ together and denote the result by $2U$: $2U=U+U$. Sample $S_c$ is drawn from one copy of the population and $S_r$ is drawn from another copy. We next use the scheme of Figure \ref{fig:scheme} in our proof.

For such a setup, by the Law of Total Probability (LTP), we have:
\begin{eqnarray}\label{eq:S_c}
   && P\left\{ i\in S_c|i \in 2U \right\} \\
   && =  P\left\{ i\in S_c|i \in U \right\}P\{i \in U \mid i \in 2U\} =\frac{1}{2}\pi_{ci}. \nonumber
\end{eqnarray}


Similarly,
\begin{eqnarray}\label{eq:S_r}
&& P\left\{ i\in S_r|i \in 2U \right\} \\ 
&& =P\left\{ i\in S_r|i \in U \right\}P\{i \in U \mid i \in 2U\}=\frac{1}{2}\pi_{ri}. \nonumber
\end{eqnarray}

The total probability of being included into the pooled sample is simply the sum of probabilities for the non-overlapping sets: 
\begin{eqnarray}\label{eq:S}
&& P\left\{ i\in S|i \in 2U \right\} \\
&& = P\left\{ i\in S_c|i \in 2U \right\}  + P\left\{ i\in S_r|i \in 2U \right\} \nonumber \\
&&=\frac{1}{2}\pi_{ci}+\frac{1}{2}\pi_{ri}. \nonumber
\end{eqnarray}

Finally, by the definition of conditional probability, 
\begin{align}\label{eq:S_c|S}
P\left\{ i\in S_c|i\in S,i \in 2U \right\} &= \frac{P\left\{ i\in S_c|i \in 2U \right\}}{P\left\{ i\in S|i \in 2U \right\}}.
\end{align}
Formula (\ref{eq:pi_z}) directly follows from putting (\ref{eq:S_c}) and (\ref{eq:S}) in formula (\ref{eq:S_c|S}). \\
\begin{figure}
\centering
 \includegraphics[width=0.7\linewidth]{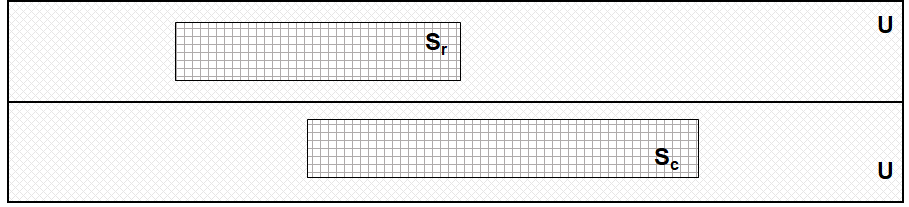}
 \caption{In this setup, two identical copies of target population $U$ are stacked together, so that $2U=U+U$.  Gridded area represents observed convenience $S_c$ and reference $S_r$ samples. Convenience sample $S_c$ is selected from one copy of population $U$ and reference sample $S_r$ is selected from another copy of $U$. Stacked together, the samples form combined sample $S$: $S=S_c+S_r$; under this scheme, if samples $S_c$ and $S_r$ overlap, the overlapping units are included in $S$ twice.}
  \label{fig:scheme}
\end{figure}
\end{proof}


\begin{remark}
In the approach of \citet{2009elliot} and \citet{2017elliot} expression (\ref{eq:pi_z}) 
is an \emph{approximation} under the assumption of a negligible overlap between samples $S_c$ and $S_r.$ Note, that unless sampling inclusion probabilities are small, the probability of a unit being included in both samples is not trivial and should be accounted for. Under our setup, this expression \emph{exactly} holds without assuming a small overlap. 
\end{remark}

\begin{remark}
The stacking devise given in Figure \ref{fig:scheme} is useful not only as a ``trick'' for the derivation of formula (\ref{eq:pi_z}). In usual practical situations, samples come from different sources  and we \emph{do not observe} which units belong to both the reference and non-probability samples and thus we would simply stack the two samples together without matching. Accounting for the probability of the overlap would contradict with such a setup.  However, if information on the identity of the overlapping units is available to researches, then it would be advisable to use this additional information and consider the union, $S_c\cup S_r,$ of the two samples instead of the stacked sample setup. 

{If indicator $I_z$ is defined on the union of two samples as having the value of 1 for units belonging to $S_c$ and 0 for the rest of the union, $S_r \setminus S_c$, the corresponding formula for  $\pi_{zi}$ would take the form:}

\begin{equation}\label{eq:identityILRwithO}
\pi_{zi}= \frac{\pi_{ci}}{\pi_{ci}+\pi_{ri}-\pi_{ci}\pi_{ri}}. 
\end{equation}

{ Such a definition of $I_z$ would correspond to obtaining the union $S_c\cup S_r$ from the stacked samples by removing the overlapping units from the reference sample. If we decide to remove the overlapping units from the non-probability sample and re-define $I_z$ as taking the value of 1 for $S_c \setminus S_r$, and 0 for $S_r$, the formula for  $\pi_{zi}$ corresponding to such defined $I_z$ would take the form: } 

\begin{eqnarray}
\pi_{zi}= \frac{\pi_{ci}\left( 1- \pi_{ri}\right)}{\pi_{ci}+\pi_{ri}-\pi_{ci}\pi_{ri}} \label{eq:identityILRwithO_Sc}.
\end{eqnarray}
{ We do not further consider either of these setups in the current paper.} 
\end{remark}

\begin{remark}
We could generalize the set-up in Section~\ref{sec:setup} to allow each of the convenience and reference samples to participate or be selected from respective frames that are \emph{subsets} of target population $U$. Namely, we may let the non-probability sample be realized from frame $U_c \subseteq U$  and probability sample be selected from frame $U_r \subseteq U$ with respective inclusion probabilities $\pi_{ci}=P\{i \in S_c|i \in U_c\}$ and $\pi_{ri}=P\{i \in S_r|i \in U_r\}$. 
Define $p_{ci}=p_{c}(\mathbf{x}_i)=P\{i\in U_c|i \in U, \mathbf{x}_i\}$ and $p_{ri}=p_{r}(\mathbf{x}_i)=P\{i\in U_r|i \in U, \mathbf{x}_i\}$ to be coverage probabilities. A coverage probability denotes the probability that a unit from the target population is included in the frame from which sample units participate. 

In order to ensure that realized frames are representative of target population, $U$, we assume that both coverage probabilities are strictly positive, $p_{ci} > 0, p_{ri} >0$. 


This assumption of positive coverage probabilities may be violated in the case of hard-to-reach subgroups in the population.  Practitioners should carefully examine the correlation between  population characteristics {of interest}, on the one hand, and covariates,  $\mathbf{x}_i$, used to define coverage probabilities, on the other hand.  If some subgroups are known, by design, to have $0$ probability to participate in the convenience sample, then \citet{2020Valliant} recommend to exclude these subgroups from the studied population.

Formula \eqref{eq:pi_z} in this case would be generalized to
\begin{align} \label{eq:pi_z_with_p}
{\pi_{zi}}=\frac{{\pi_{ci}}p_{ci}}{{\pi_{ci}}p_{ci}+{\pi_{ri}}p_{ri}}.
\end{align}

For example, suppose household units in $U_c \subseteq U$ have broadband internet access and can potentially participate in a web survey. Assume, each household in the target population has positive probability $p_{ci}=P\{i\in U_c|i \in U,\mathbf{x}_i\}$ to have broadband internet access and let this probability be \emph{known} for some broadly defined demographic groups $\mathbf{x}_i.$ At the same time, survey participation probabilities $\pi_{ci}=P\{i\in S_c|i \in U_c,\mathbf{x}_i\}$ are \emph{not known}. Considering the reference sample case, suppose each person in the target population has positive probability to have a mailable address and let this probability $p_{ri}=P\{i\in U_r|i \in U,\mathbf{x}_i\}$ be known for a given demography. Frame $U_r$ is a subset of $U$ consisting of persons who actually have mailable addresses and thus could be selected into the reference sample. In this example, $p_{ci}$ and $p_{ri}$ play the role of known coverage probabilities. 

We make brief mention that CRISP formula \eqref{eq:pi_z_with_p} may also be used in other sampling situations. 
For example, we may define $(p_{ri},p_{ci})$ as first stage participation probabilities in a multi-stage sampling design (rather than viewing them as coverage probabilities) to construct a sampling frame.  In this set-up, $(\pi_{ri},\pi_{ci})$, are second stage inclusion probabilities for units selected from the realized frame.  The first stage probabilities 
may be \emph{unknown}, but once the sampling frame is defined, the second stage inclusion probabilities, by contrast, can be set and \emph{known}.  In this set-up, the focus for inference would be the first stage participation probabilities in the sampling frame, $(p_{ri},p_{ci})$, rather than $(\pi_{ri},\pi_{ci})$.

One class of examples of such a multi-stage participation and selection design belongs to the opt-in panels recruited over the internet. The probability of signing up for a web-based panel (``the first stage'') is not known. On ``the second stage'', once the panel has been recruited, a sample -- with \emph{known} inclusion probabilities -- would be selected for a given survey. 

See also related discussion in \citet{2017elliot} and \citet{2020Valliant}.
In the remainder of the paper, we retain the assumption that both convenience and reference samples are both drawn from the same population frames representing population $U$ for simplicity and without loss of generality.
\end{remark}

Let us now consider joint inclusion probabilities of units $i,j \in U$, $i \neq j$, conditional on covariates. Denote 
\begin{align*}
\pi_{c,ij}=\pi_{c}(\mathbf{x}_i,\mathbf{x}_j)&=P\{i,j\in S_c|i,j \in U_c,\mathbf{x}_i,\mathbf{x}_j\},\\ 
\pi_{r,ij}=\pi_{r}(\mathbf{x}_i,\mathbf{x}_j)&=P\{i,j\in S_r|i,j \in U_r,\mathbf{x}_i,\mathbf{x}_j\}.
\end{align*}

The joint conditional independence of inclusion indicators $I_{ci}$ and $I_{cj}$ stated in Assumption~\ref{itm:indepI_c} is a serious supposition that, nevertheless, is reasonable given conditioning on a rich set of covariates. The same reasoning can be applied to the assumption of conditional independence of the reference sample inclusion indicators $I_{ri}$ and $I_{rj}$ given $\mathbf{x}_i$ and $\mathbf{x}_j.$

We now consider joint probabilities $\pi_{z,ij}=\pi_{z}(\mathbf{x}_i,\mathbf{x}_j)=P\{i,j\in S_c|i,j \in S,\mathbf{x}_i,\mathbf{x}_j\}$  of being in the convenience sample for different units $i \neq j$ in stacked sample $S$. The problem is that a pair of  units in $S$  may represent the same population unit that belongs to both the convenience and reference samples. There is a common understanding in the literature that such overlapping units are correlated and have to be taken under account; see, for example, discussion in \citet{2009elliot},  \citet{2017elliot},   \citet{2021KimTam}, \citet{2022LiuScholtusWaal}. However, identifying overlapping units may be a substantial undertaking since population labels usually are not available and it is often difficult to distinguish units that appear twice in the combined set. The following Theorem implies that, conditional on covariates, indicators $I_{z}$  on staked sample $S$ are nearly independent even when the sampling overlap is large. This Theorem is important since it allows us to use the likelihood for independent Bernoulli variables formulated over all units in the stacked sample.



\begin{theorem}
\label{theo:2_2}
Assume, inclusion indicator variables in each sampling stage are mutually independent conditional on covariates.
Then, conditional on these covariates, covariance between indicator variables $I_{zi}$ and $I_{zj}$ on staked sample $S$ is
\begin{align}\label{eq:pi_zij}
Cov(I_{zi},I_{zj})=-O\left(N^{-1}\right).
\end{align}
\end{theorem}
The proof of the Theorem \ref{theo:2_2} is given in Appendix \ref{sec:ProofT2.2}. 

\begin{remark}
 Note that indicators $I_{z}$ are nearly independent regardless of the size of the sampling overlap. For example, in the extreme case of the ultimate overlap where the whole population is selected into both samples, the correlation between $I_{z}$ is $-1/(2N-1)$, which is small for any population size of a practical matter.
\end{remark}

\subsection{Modeling using CRISP based on valid likelihood for the observed sample} \label{sec:ILR_PILR}

\citet{beresovsky2019} proposed using CRISP expression \eqref{eq:pi_z} to obtain estimates of $\pi_{ci}$ directly from a likelihood formulated on the combined sample. This approach is labeled ``Implicit Logistic Regression (ILR)'' because parameters are entering the logistic link function ``implicitly'' through a composite function.


Participation probabilities $\pi_{ci}=\pi_{ci}(\boldsymbol{\beta})$ are parameterized as in the CLW approach by $\text{logit}(\pi_{ci}(\boldsymbol{\beta}))=\boldsymbol{\beta^T}\mathbf{x}_i,$  and identity (\ref{eq:pi_z}) is used to present $\pi_{zi}$ as a composite function of $\boldsymbol{\beta}$; that is, $\pi_{zi}=\pi_{zi}(\pi_{ci}(\boldsymbol{\beta}))=\pi_{ci}(\boldsymbol{\beta})/(\pi_{ri}+\pi_{ci}(\boldsymbol{\beta})).$ 
The log-likelihood for observed Bernoulli variable $I_{zi}$ is
\begin{eqnarray}\label{eq:ilr1}
    && \ell^{ILR}(\boldsymbol{\beta}) = \mathop{\sum}_{i \in S_c}\log\left[{\pi_{zi}(\pi_{ci}(\boldsymbol{\beta}))}\right] + \\ 
     && \mathop{\sum}_{i \in S_r}\log\left[1-\pi_{zi}\left(\pi_{ci}(\boldsymbol{\beta})\right)\right]. \nonumber
\end{eqnarray}


{This approach may be labeled ``indirect'' estimation because $\pi_{c}$ is indirectly used through the CRISP in Bernoulli likelihood  formulated for the \emph{observed} indicators $I_{zi}$ of convenience sample on the ``stacked'' set $S=S_c + S_r$.}

\begin{remark} \label{theo:r_indepI_z}
Note, by Theorem \ref{theo:2_2}, indicators $I_{zi}$ are practically \emph{independent} for units on $S$, even if samples $S_c$ and $S_r$ have common units. We use exact likelihood for independent Bernoulli variables. Thus, standard statistical properties and techniques can be applied. For example, we can apply usual likelihood based model selection methods such as AIC or BIC.   
\end{remark}

The score equations are obtained from (\ref{eq:ilr1}) by taking the derivatives, with respect to $\boldsymbol{\beta},$ of composite function $\pi_{zi}=\pi_{zi}(\pi_{ci}(\boldsymbol{\beta})).$ This way, the estimates of $\pi_{ci}$ are obtained directly from (\ref{eq:ilr1}) in a single step. 
The resulting score equations are
\begin{eqnarray}\label{eq:ilr_score}
& \boldsymbol{0} =  S^{ILR}(\boldsymbol{\beta})= \frac{\partial \ell^{ILR}(\boldsymbol{\beta})}{\partial \boldsymbol{\beta}} = \\
& \mathop{\sum}_{i \in S_c}(1-\pi_{zi})(1-\pi_{ci})\mathbf{x}_i   
 -\mathop{\sum}_{i \in S_r}\pi_{zi}(1-\pi_{ci})\mathbf{x}_i. \nonumber
\end{eqnarray}

\subsection{Pseudo-ILR method} \label{sec:PILR}

\citet{2021valliant} also developed an indirect method for estimation of $\pi_{ci}$ as a function of $\pi_{zi}$, which we reveal is a special case of \eqref{eq:pi_z}. For their Adjusted Logistic Propensity (ALP) weighting method, \citet{2021valliant} use a stacking approach as do we.  They introduce an imaginary construct consisting of two parts: they \emph{stack} together non-probability sample $S_c$ (part 1) and finite population $U$ (part 2). Since non-probability sample units belong to the finite population, they appear in the stacked set twice. They formulate a Bernoulli likelihood for variable $I_{{\delta}i}$, where $I_{{\delta}i}=1$ if unit $i$ belongs to part 1 of the stacked set; and $I_{{\delta}i}=0$, otherwise: 
\begin{eqnarray}\label{eq:wang1}
    &&\ell^{ALP}(\boldsymbol{\gamma}) = \\
    && \mathop{\sum}_{i \in S_c}\log\left[{\pi_{{\delta}i}(\boldsymbol{\gamma})}\right] +  \mathop{\sum}_{i \in U}\log\left[1-\pi_{{\delta}i}\left(\boldsymbol{\gamma}\right)\right], \nonumber
\end{eqnarray}
where $\boldsymbol{\gamma}$ is the parameter vector in a logistic regression model for $\pi_{{\delta}i}(\mathbf{x}_i)=P\{I_{{\delta}i}=1|\mathbf{x}_i\}.$
At the second step of the ALP method, estimates of $\pi_{ci}$ are derived from identity
\begin{equation}\label{eq:identityWVL}
\pi_{{\delta}i}= \frac{\pi_{ci}}{\pi_{ci}+1}.   
\end{equation}
Their identity may also be derived from our result in \eqref{eq:pi_z} by replacing $\pi_{ri}$ with $1$ because the second part of the
stacked sample is the finite population. 

Since the finite population is not available, they apply a pseudo-likelihood approach:  
\begin{eqnarray}\label{eq:wang2}
    && \hat{\ell}^{ALP}(\boldsymbol{\gamma}) = \\
    && \mathop{\sum}_{i \in S_c}\log\left[{\pi_{{\delta}i}(\boldsymbol{\gamma})}\right] + \mathop{\sum}_{i \in S_r}w_{ri}\log\left[1-\pi_{{\delta}i}\left(\boldsymbol{\gamma}\right)\right], \nonumber
\end{eqnarray}
leading to an estimate of $\pi_{{\delta}i}.$ 
However, the actual goal is to find probabilities $\pi_{ci}$ rather than $\pi_{{\delta}i}.$ On the second step of the ALP approach, \citet{2021valliant} obtain the estimate of $\pi_{ci}$ using \eqref{eq:identityWVL}.

\citet{2021valliant} noted that in their simulations the ALP estimator was more efficient than the CLW, especially in the case of large non-probability and small probability sample sizes.

The estimation method of \citet{2021valliant} can be modified to a one-step estimation procedure similar to ILR: $\pi_{ci}$ can be parameterized using the logistic link function as $\text{logit}(\pi_{ci}(\boldsymbol{\beta}))=\boldsymbol{\beta^T}\mathbf{x}_i$, while probabilities $\pi_{{\delta}i}$ in (\ref{eq:identityWVL}) could be viewed as a composite function, $\pi_{{\delta}i}=\pi_{{\delta}i}(\pi_{ci}(\boldsymbol{\beta}))=\pi_{ci}(\boldsymbol{\beta})/(1+\pi_{ci}(\boldsymbol{\beta}))$.     
We call this method pseudo-ILR (PILR), and the pseudo-likelihood is 
\begin{eqnarray}\label{eq:wang3}
    && \hat{\ell}^{PILR}(\boldsymbol{\beta}) = 
    \mathop{\sum}_{i \in S_c}\log\left[{\pi_{{\delta}i}(\pi_{ci}(\boldsymbol{\beta}))}\right] + \\
    && \mathop{\sum}_{i \in S_r}w_{ri}\log\left[1-\pi_{{\delta}i}(\pi_{ci}(\boldsymbol{\beta}))\right]. \nonumber
\end{eqnarray}
The score equations are:
\begin{eqnarray}\label{eq:PILR_score}
& \boldsymbol{0} = S^{PILR}(\boldsymbol{\beta})=\frac{\partial \hat{\ell}^{PILR}(\boldsymbol{\beta})}{\partial \boldsymbol{\beta}}=  \\
&  \mathop{\sum}_{i \in S_c}(1-\pi_{{\delta}i})(1-\pi_{ci})\mathbf{x}_i - \mathop{\sum}_{i \in S_r}w_{ri}\pi_{{\delta}i}(1-\pi_{ci})\mathbf{x}_i. \nonumber
\end{eqnarray}

This change in estimation of model parameters makes the approach more efficient and it avoids cases where estimates of $\pi_{ci}$ become greater than $1$, as may occur under the ALP approach where the estimation is performed in two steps.
{See plots of predicted vs true participation probabilities in Figure \ref{fig:step1vs2} of Appendix \ref{sec:step1_vs_step2}.}

\begin{remark}
Similar to Remark \ref{theo:r_indepI_z}, by Theorem \ref{theo:2_2}, indicators $I_{\delta i}$ and $I_{\delta j}$ are nearly independent given $\mathbf{x}_i$ and $\mathbf{x}_j$ for $i \neq j$ in the stacked $S_c$ and $U$ set, even though each $S_c$ unit is included into the stacked set twice, once with $I_{\delta i}=1$ and once with $I_{\delta i}=0.$ Thus, the likelihood for independent Bernoulli variables is well justified for the stacked $S_c$ and $U$ set.
\end{remark}

\begin{remark}
In any of these approaches, there is no need to restrict the parameterization to a linear form and one could parameterize the likelihood as $\text{logit}(\pi_{ci}(\boldsymbol{\beta}))=\psi(\mathbf{x}_i; \boldsymbol{\beta})$ for some general function $\psi.$ For example, \citet{savitsky2023} used splines to estimate the model parameters. 
\end{remark}

\begin{remark}
    In the sample-based approach of Section \ref{sec:ILR_PILR}, reference sample inclusion probabilities $\pi_r$ are assumed to be known for all units in combined sample $S$: for both the probability and non-probability parts of the set. In practice, $\pi_r$'s are often not available for units outside the probability sample. For some designs (for example, in stratified samples), reference sample inclusion probabilities can be easily inferred for all units. \citet{2017elliot} suggest to model $\pi_r$ first, then plug the estimates of $\pi_r$ into the formula; \citet{savitsky2023} use a Bayesian formulation to co-model unknown $\pi_r$ along with $\pi_c$, based on available covariates. Thus, if $\pi_r$'s are available for all units, the sample-based likelihood approach allows to conveniently use them in estimation; if $\pi_r's$ are not available, we can incorporate covariates to derive them. In contrast, pseudo-likelihood approaches do not require $\pi_r's$ outside the reference sample. At the first glance, this may seam like an advantage of the pseudo-likelihood formulations; however, by the same token, the downside is that pseudo-likelihood formulations do not anticipate using this information when it is available, and would have to ignore it.
\end{remark}

\subsection{Inverse probability weighted estimator} \label{sec:ipw}

In the sequel we consider the following H\'{a}jek form of the inverse probability weighted estimator of finite population mean $\mu$:
\begin{equation}\label{eq:hajek}
    \hat{\mu} =\frac{1}{\hat{N}} \mathop{\sum}_{i \in S_c}\frac{y_i}{\hat{\pi}_{ci}},
\end{equation}
where $\hat{N}=\sum_{i \in S_c}(\hat{\pi}_{ci})^{-1}$ and $\hat{\pi}_{ci}$ is an estimate of ${\pi}_{ci}$.

Note that the Hájek estimator of parameter $\mu$ can be presented as a solution of the following estimating equation:
\begin{equation}\label{eq:hajek_est}
    U(\mu)=\mathop{\sum}_{i \in U}I_{ci}\hat{\pi}_{ci}^{-1}(y_i-\mu)=0. 
\end{equation}

\section{Theoretical properties} \label{sec:theorems}

Theoretical properties of the CLW estimators are given in \citet{2020_ChenLiWu}. In Section \ref{sec:asymvars}, we present asymptotic variances of the estimated parameters and rates of convergence to true values for both of the ILR and PILR.  We also include the CLW variance to facilitate our comparison of the relative efficiencies of these methods. In Section \ref{sec:theorstudy}, we discuss the expected behavior of the asymptotic variances for ILR, PLR and CLW. Section \ref{sec:numstudy} presents a numerical study of relative efficiency of the estimators for the methods over a sequence of sampling fractions.

\subsection{Asymptotic variances of CLW, ILR, and PILR estimators} \label{sec:asymvars}

We study asymptotic properties of the estimators using the following framework. Consider a sequence of finite populations $U_{\nu}$, indexed by $\nu$, of increasing sizes $N_{\nu}$, so that $N_{\nu} \to \infty$ as $\nu \to \infty,$ and let non-probability sample $S_{c\nu}$ of size $n_{c\nu}$ and reference probability sample $S_{r\nu}$ of size $n_{r\nu}$ be independently drawn from population $U_{\nu}.$  The probability sample design is assumed to be the same for each $U_{\nu}$ and assume that both sample sizes increase $n_{c\nu} \to \infty$ and $n_{r\nu} \to \infty$ as $\nu \to \infty$ in such a way that sampling fractions $f_c=\lim_{\nu \to \infty}n_{c\nu}/N_\nu$ and $f_r=\lim_{\nu \to \infty}n_{r\nu}/N_\nu.$ The same limiting process in considered by \citet{2020_ChenLiWu}.
As is customary, to simplify notation, we drop index $\nu$ in the rest of the paper.  

We first present the asymptotic variances of the estimates of parameters $\mu$ and $\boldsymbol{\beta}$ in a general form applicable for all three methods. Such a representation aids in comparing methods: 
\begin{eqnarray}
  && Var(\hat{\mu}) \doteq   N^{-2}\{Var[U(\mu)] \label{eq:varmu_general}\\ 
  && -2\boldsymbol{b}^TCov[U(\mu),{S}(\boldsymbol{\beta})]+  \boldsymbol{b}^TVar[{S}(\boldsymbol{\beta})]\boldsymbol{b}\}, \nonumber \\
 && Var(\boldsymbol{\hat{\beta}}) \doteq \boldsymbol{H}^{-1}Var[{S}(\boldsymbol{\beta})]\boldsymbol{H}^{-1} \label{eq:varbeta_general},
\end{eqnarray}
where 
$U(\mu)$ is the score function for $\mu$, ${S}(\boldsymbol{\beta})$ are the score functions for parameters $\boldsymbol{\beta}$ in a given model for $\pi_{ci}$, 
and $\boldsymbol{b}=S_{\boldsymbol{\beta}}^{-1}U_{\boldsymbol{\beta}}^T$, 
$U_{\boldsymbol{\beta}} = E[{\partial U(\mu)}/{\partial \boldsymbol{\beta}^T}]$, 
$S_{\boldsymbol{\beta}} = E[{\partial S(\boldsymbol{\beta})}/{\partial \boldsymbol{\beta}}]$, $\boldsymbol{H}=-S_{\boldsymbol{\beta}}.$ 
See Appendix \ref{sec:ProofT5.1} for the derivation.

The first term in \eqref{eq:varmu_general},  $Var[U(\mu)],$ is the design-based variance under the probability sampling design for the convenience sample, and the other two terms are related to the estimation of unknown participation probabilities $\pi_{ci}$ involved in the probability weighted estimator for $\mu.$ Importantly, (\ref{eq:varmu_general}) contains variance $Var[{S}(\boldsymbol{\beta})]$ due to estimation of parameters $\boldsymbol{\beta},$  which shows that efficiency of estimation of $\mu$ depends on the quality of the model for $\pi_{ci}.$ The score function for $\boldsymbol{\beta}$ can be further decomposed into two independent parts, corresponding to the non-probability and reference samples: 
\begin{align}
{S}(\boldsymbol{\beta})={S_c}(\boldsymbol{\beta})+{S_r}(\boldsymbol{\beta}), 
\end{align}
and the corresponding variance is   
\begin{eqnarray}
&&Var[S(\boldsymbol{\beta})]= \\
&&Var[S_c(\boldsymbol{\beta})]+Var[S_r(\boldsymbol{\beta})]=:\mathbf{A}+\mathbf{D}. \nonumber
\end{eqnarray}
Unless the reference sample is large, the design variance $\mathbf{D}=Var[S_r(\boldsymbol{\beta})]$ is a significant contributor into the overall variance. In Section  \ref{sec:theorstudy} below, we discuss and compare the effect of the design-based variance for each of the three methods.

We now follow with a formal statement that spells out detailed formula for each of the methods, CLW, ILR, and PILR. Assume the following regularity conditions:
\begin{enumerate}
    \item [C1] Matrix $\sum_{i \in U}\mathbf{x}_i\mathbf{x}_i^T$ is positive-definite.
    \item [C2] There exist $c_2 \ge c_1 >0$ such that $c_1 \le \pi_{ci}N/n_c\le c_2$ and $c_1 \le \pi_{ri}N/n_r\le c_2$ for all units $i.$
\end{enumerate}

\begin{theorem}[Consistency of the estimators]\label{theo:consist}
Assume logistic regression parameterization for the non-probability sample participation probabilities, 
\\
$\text{logit}\left[\pi_{ci}(\boldsymbol{\beta})\right]=\boldsymbol{\beta^T}\mathbf{x}_i.$  Then, under regularity conditions C1-C2, estimators $\hat{\boldsymbol{\eta}}_m=(\hat{\mu}_m,\hat{\boldsymbol{\beta}}_m)$ for $\boldsymbol{\eta}=({\mu},{\boldsymbol{\beta}})$ based on CLW, ILR, or PILR methods are design consistent, $\hat{\boldsymbol{\eta}}_m-\boldsymbol{\eta}=O_p(\text{min}(n_c,n_r)^{-1/2})$, and variances are 
\begin{eqnarray}
   && Var(\hat{\mu}_m) \doteq N^{-2}  
    \mathop{\sum}_{i \in U} \frac{1-\pi_{ci}}{\pi_{ci}}(y_i-\mu)^2    
      \label{eq:varmuM}  \\
   &&  - 2 N^{-2}\boldsymbol{b}_m^T\boldsymbol{C}_m + 
N^{-2}\boldsymbol{b}_m^T(\boldsymbol{A}_m + \boldsymbol{D}_m)\boldsymbol{b}_m  , \nonumber  \\
   && Var(\boldsymbol{\hat{\beta}}_m) \doteq \boldsymbol{H}_m^{-1}
   (\boldsymbol{A}_m+\boldsymbol{D}_m)\boldsymbol{H}_m^{-1} \label{eq:varbetaM},
\end{eqnarray}

where  subscript $m$ stands for methods (CLW, ILR, or PILR), $\boldsymbol{b}_m=\boldsymbol{H}_m^{-1}\sum_{i \in U}(1-\pi_{ci})(y_i-\mu)\mathbf{x}_i,$ and the other terms are spelled out as follows:

For the CLW method,
\begin{align*}
\boldsymbol{C}_{CLW}&=\sum_{i \in U}(1-\pi_{ci})(y_i-\mu)\mathbf{x}_i,\\
\boldsymbol{H}_{CLW}&=\boldsymbol{A}_{CLW}=\sum_{i \in U}\pi_{ci}(1-\pi_{ci})\mathbf{x}_i\mathbf{x}_i^T,
\end{align*}
and
\begin{align*}
\boldsymbol{D}_{CLW}=Var_d\left[\sum_{i \in S_r}w_{ri}\pi_{ci}\mathbf{x}_i\right]
\end{align*}
is the design-based variance-covariance matrix under the probability sampling design for $S_r$.

For the ILR method,
\begin{eqnarray*}
 && \boldsymbol{C}_{ILR}=  
  \mathop{\sum}_{i \in U}(y_i-\mu)(1-\pi_{zi})(1-\pi_{ci})^2\mathbf{x}_i,  \nonumber \\
  && \boldsymbol{H}_{ILR}=   \sum_{i \in U} (\pi_{ci}+\pi_{ri}) 
   \pi_{zi}(1-\pi_{zi})(1-\pi_{ci})^2\mathbf{x}_i\mathbf{x}_i^T, \nonumber\\ 
  && \boldsymbol{A}_{ILR}=\boldsymbol{H}_{ILR}-  \sum_{i \in U}(\pi_{ri}+1) \times \\
  && \pi_{ci}\pi_{zi}(1-\pi_{zi})(1-\pi_{ci})^2\mathbf{x}_i\mathbf{x}_i^T, \nonumber 
\end{eqnarray*}
and 
\begin{align*}
\boldsymbol{D}_{ILR}=Var_d\left[\mathop{\sum}_{i \in S_r}\pi_{zi}(1-\pi_{ci})\mathbf{x}_i \right]
\end{align*}
is the design-based variance-covariance matrix under the probability sampling design for $S_r$.

For the PILR method,
\begin{eqnarray*}
&& \boldsymbol{C}_{PILR} =  
 \mathop{\sum}_{i \in U}(y_i-\mu)(1-\pi_{{\delta}i})(1-\pi_{ci})^2\mathbf{x}_i, \nonumber \\
&& \boldsymbol{H}_{PILR} =  
 \sum_{i \in U}(\pi_{ci}+1)\pi_{{\delta}i}(1-\pi_{{\delta}i})(1-\pi_{ci})^2\mathbf{x}_i\mathbf{x}_i^T, \nonumber\\ 
&& \boldsymbol{A}_{PILR} = \boldsymbol{H}_{PILR} \\
&& -2\mathop{\sum}_{i \in U}\pi_{ci}\pi_{{\delta}i}(1-\pi_{{\delta}i})(1-\pi_{ci})^2\mathbf{x}_i\mathbf{x}_i^T, \nonumber
\end{eqnarray*}
and
\begin{eqnarray*}
&& \boldsymbol{D}_{PILR}= 
 Var_d\left[\sum_{i \in S_r}w_{ri}\pi_{ci}(1-\pi_{ci})(1-\pi_{{\delta}i})\mathbf{x}_i\right] \nonumber
\end{eqnarray*}
is the design-based variance-covariance matrix under the probability sampling design for $S_r$. 
\end{theorem}

The details of the proof for ILR and PILR are given in Appendix \ref{sec:ProofT5.1}.


\subsection{Analysis of the relative efficiency of the methods} \label{sec:theorstudy}
In this Section, we take a closer look at the components of variances presented in Section \ref{sec:asymvars}. 

Consider a situation where the whole population $U$ is available in place of a reference sample (that is,  sampling fraction $f_r=n_r/N =1$).  Since Bernoulli likelihood (\ref{eq:chen1})  used in the CLW method describes the true data-generating process at the population level,  the CLW approach is the preferable one for estimation  in the case where the whole $U$ is observed. Under this scenario, design-based variance-covariance matrices $\boldsymbol{D}_{CLW},$ $\boldsymbol{D}_{ILR}$, and $\boldsymbol{D}_{PILR}$  are reduced to zero matrices, reflecting the fact that variability due to the probability sampling does not contribute to respective total variances of parameters  ${\mu}$ and ${\boldsymbol{\beta}}.$ In particular, the variance of parameter $\boldsymbol{\beta}$ given by (\ref{eq:varbetaM}) would be reduced to $\boldsymbol{H}_{CLW}^{-1}$ and equal to the  variance  of the maximum likelihood estimator from Bernoulli likelihood (\ref{eq:chen1}).  In this special case, the ILR and PILR methods are identical  and yield estimators that are less efficient than CLW estimators. 

Let us now turn to a more common situation where probability sample $S_r$ with sampling fraction $f_r=n_r/N<1$  is available rather than the whole population $U.$  To a large degree, respective contributions of  design variances  $\boldsymbol{D}_{CLW},$ $\boldsymbol{D}_{ILR}$, and $\boldsymbol{D}_{PILR}$   determine the size of the total variances.  Consider 
\begin{align}
\boldsymbol{D}_{CLW}
&= {Var_{d}}\left[ \sum\nolimits_{S_r}{  {g_i} {{\mathbf{x}}_{i}}} \right], \label{eq:D.CLW} \\
     \boldsymbol{D}_{ILR}  
     &= 
    {Var_{d}}\left[ \sum\nolimits_{S_r}{ {\frac{ g_i}{ 1 + g_i }}  \left(1-{{\pi }_{ci}} \right) {{\mathbf{x}}_{i}}} \right], \label{eq:D.ILR}\\
 \boldsymbol{D}_{PILR} 
 &= {Var_{d}}\left[ \sum\nolimits_{S_r}{  {g_i}  \frac{ 1-{{\pi }_{ci}} }{ 1+{{\pi }_{ci}} }  {{\mathbf{x}}_{i}}} \right], \label{eq:D.PILR} 
\end{align}
where $g_i=g\left( \mathbf{x}_i\right) ={\pi_{c}\left( \mathbf{x}_i\right)}/{\pi_{r}\left( \mathbf{x}_i\right)} \in \left(0, \infty \right)$. Variation in ratio  $g_i$ depends on the degree of common support, or overlap, between two samples in the covariate-defined domains. In the case of ``high overlap'', $\pi_{ci}$ and $\pi_{ri}$ would have ``similar dependence'' on $\mathbf{x}_i$, and there will be relatively minor fluctuations of ratios $g_i$. 
However, in the case of ``low overlap'', $\pi_{ci}$ and $\pi_{ri}$ would have an ``opposite dependence'' on $\mathbf{x}_i$  and there would be a significant variation of  $g_i.$ 

In context of the observational studies, \citet{Firth1993} and \citet{HeinzeSchemper2002} considered estimation of parameters of generalized linear models for the case of small overlap, or ``separation'', between treated and control units in variable-defined domains. They showed that separation in variable-defined domains 
may result in biased and unstable estimates of model parameters. 


Expressions (\ref{eq:D.CLW}) - (\ref{eq:D.PILR}) indicate that susceptibility to the effect of a low overlap in variable-defined domains  varies between the three methods. Expression (\ref{eq:D.CLW}) for  $\boldsymbol{D}_{CLW}$ 
includes $g_i$  as a factor that potentially can vary in the whole interval $\left(0, \infty \right)$.  Expression  (\ref{eq:D.PILR})  defining design variance $\boldsymbol{D}_{PILR}$  
differs from the CLW expression by factor  $(1-\pi_{ci})/(1+\pi_{ci})<1$, that could soften the effect from the variability of  $g_i$.  In expression (\ref{eq:D.ILR}) for $\boldsymbol{D}_{ILR},$ 
factor ${g_i}/(1 + g_i)$  varies between 0 and 1, thus indicating lower susceptibility of the ILR method to the low overlap problem compared to CLW or PILR. 

This quick analysis demonstrates that low overlap in covariate-defined domains may affect stability of the estimates, and the impact varies between the three methods.  

{The differences in estimated variances between the three methods are expected to be more pronounced for the situation of a large non-probability sampling fraction coupled with a small probability sampling fraction. This situation is encountered in the practical setting of a large-sized administrative data for the non-probability  (convenience) sample and a small reference sample obtained from a probability survey. In the sequel, this expectation is confirmed in numerical and simulation studies.}

\subsection{Numerical study using theoretical quantities} \label{sec:numstudy}
We present a numerical study of this influence of sampling fractions for reference and convenience samples, on the one hand, on the quality of estimates for the $3$ methods, on the other hand, by  computing  variances  for CLW,  ILR, and PILR on a grid of values of sampling fractions $\left(f_r, f_c \right)=\left(n_r/N, n_c/N \right)$ under the low and high overlap scenarios. 

For each unit $i=1\ldots,N$ in finite population $U$ of size $N=100,000$, we generate  covariate $x_i$ as independent standard normal variable, $x_i \sim N \left(0, 1 \right)$. Outcome $y_i$ is generated as a normal variable with mean depending on $x_i$ and constant variance, $y_i \sim N \left(1 + x_i, 1.5^2 \right)$.    

We use a Poisson sampling design with participation probabilities $\pi_{ci}$ to select convenience sample $S_c$ from population $U$. Probabilities $\pi_{ci}$ are generated depending on $x_i$ as
\begin{equation}\label{eq:numtrueNP}
\text{logit}(\pi_{ci})=\beta_{c0} + \beta_{c1}x_{i}.
\end{equation}
We vary values of intercept $\beta_{c0}$ to scan across a range of sampling fractions $f_c \in(0,1)$. The value of  slope $\beta_{c1}$ is set to $1.$ 

We select reference sample $S_r$ using the probability proportional to size (PPS) sampling with the measure of size $m_{ri}$ generated as
\begin{equation}\label{eq:numtrueP}
\text{logit}(m_{ri})=1+\beta_{r} x_{i}.
\end{equation}
For any size measures $m_{r} \left( x_i \right) $, overall sampling fraction $f_r$ can be tuned to a desired value using functionality of \texttt{sampling} R-package by \citet{samplingpkg2021}. 

For the low overlap scenario,  we simulate the ``opposite dependence'' of $\pi_{ci}$ and $\pi_{ri}$ on covariate $x_i$ by setting $\beta_{r}= -1. $ For the high overlap scenario, we generate ``similar dependence'' by setting $\beta_{r}= 1,$ so that  it is the same as slope  $\beta_{c1}$ in the generating model for $\pi_{ci}$.  Figure \ref{fig:overlap} shows the resulting histogram of sample densities in domains of covariate $x_i.$  Figure \ref{fig:overlap} helps to visualize and clarify what is meant by low and high overlap,  in a simple case of one covariate. In the case of multivariate models with a large number of auxiliary variables and interactions, it is likely to have very little overlap in some of the variable-defined domains, even when samples are large. A similar phenomenon  in machine learning is called ``the curse of dimensionality'' and it is recognized as a reason for unstable and not reproducible predictions.        
 
 \begin{figure}
  \centering
    \includegraphics[width = 1.0\linewidth
    ]{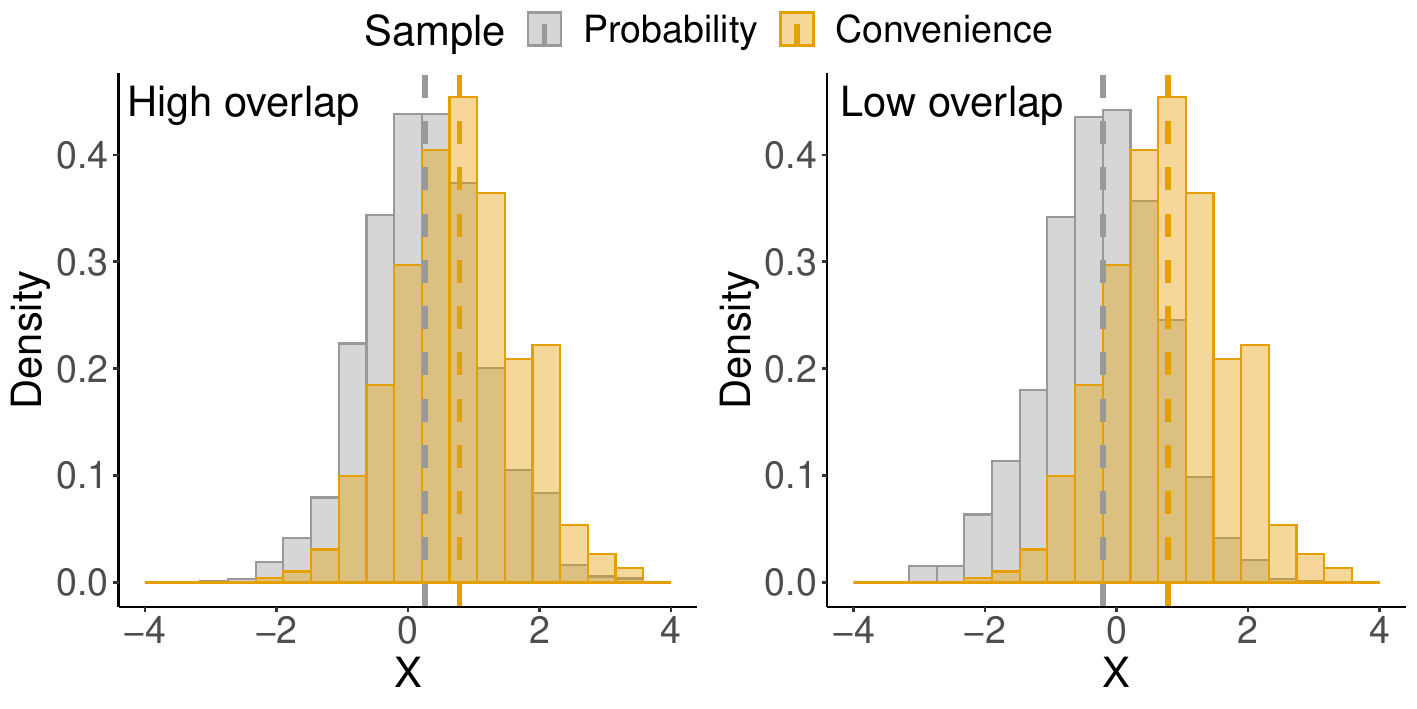}
     \caption{\small{ High overlap:  $\beta_r =  \beta_{c1}$; Low overlap: $\beta_r = - \beta_{c1}$ } } \label{fig:overlap}
\end{figure}

\begin{figure}
  \centering
    \includegraphics[width=1.0\linewidth] 
  {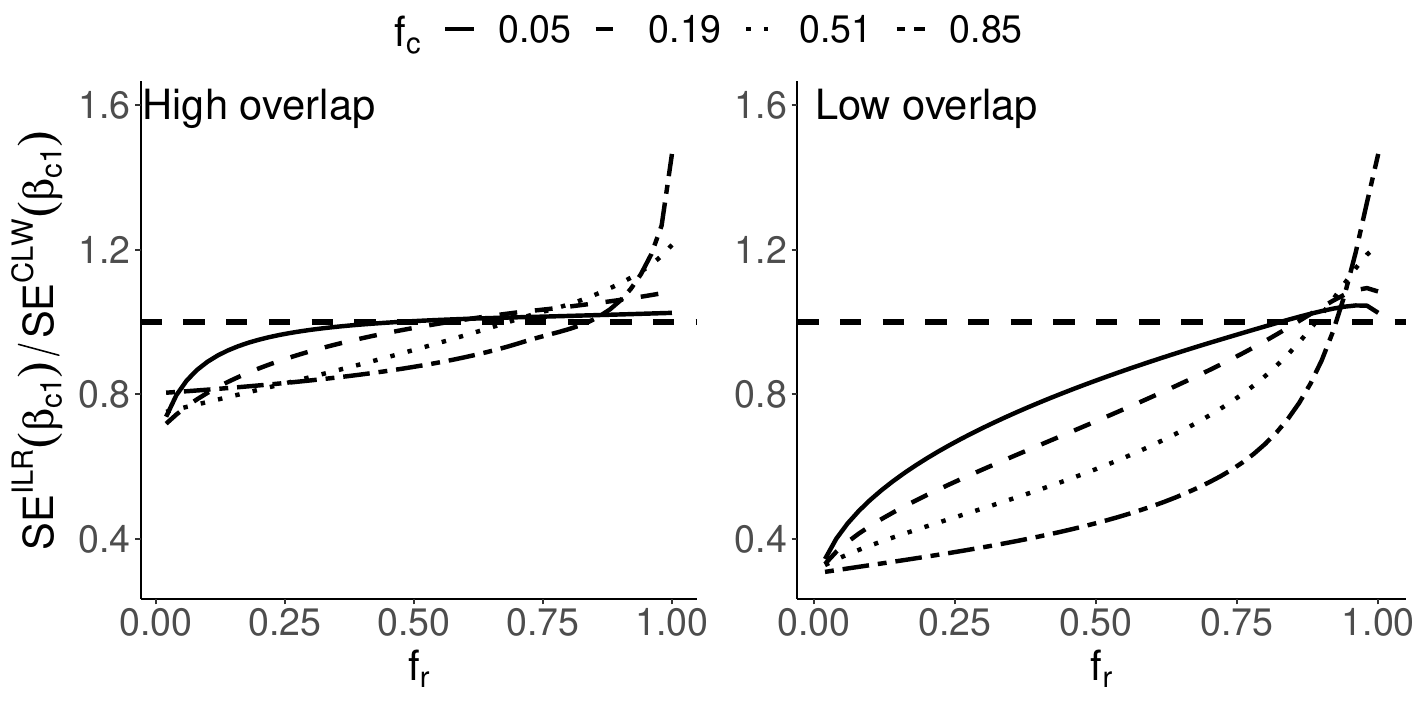}
     \caption{ Relative standard error of ILR and CLW estimates of propensity model parameter $\beta_{c1}$  } \label{fig:counterbeta}
\end{figure}

\begin{figure}
  \centering
    \includegraphics[width = 1.0\linewidth]
  {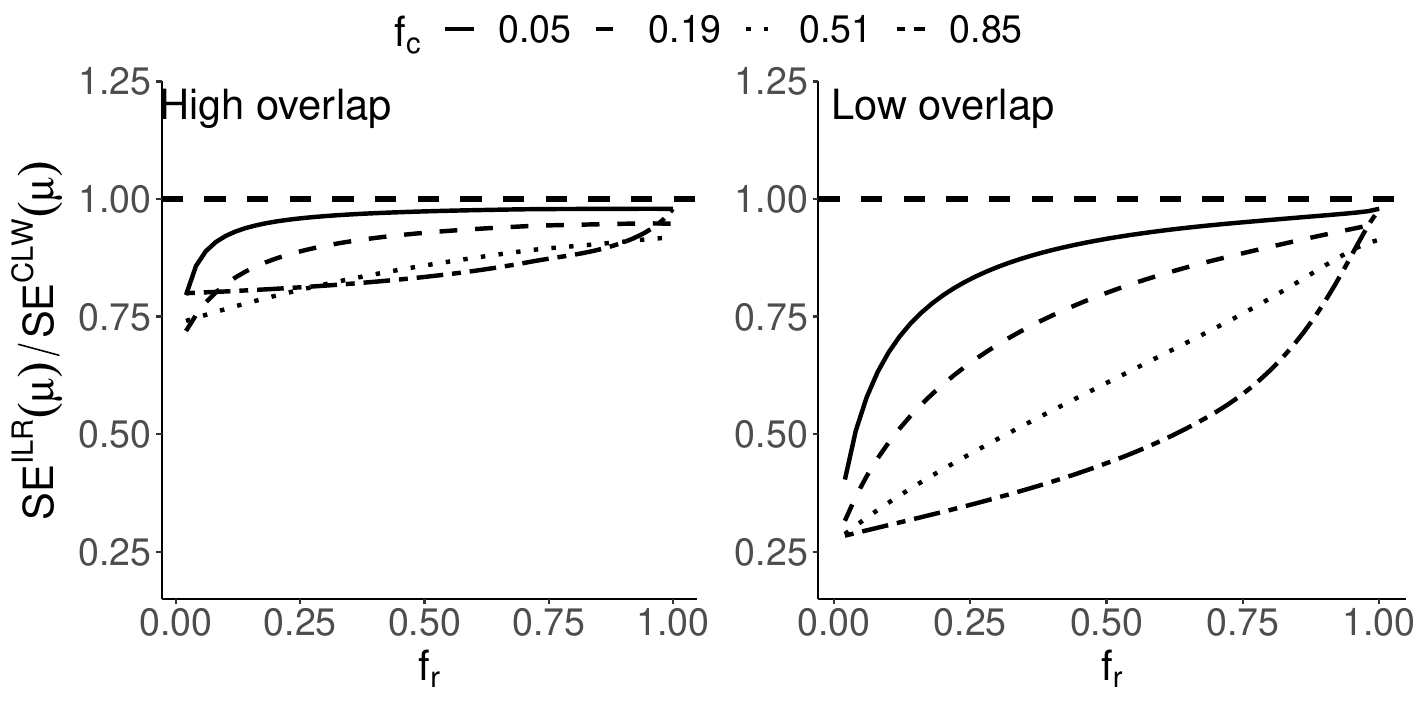}
     \caption{ Relative standard error of ILR and CLW estmates of population mean $\mu$ } \label{fig:countermu}
\end{figure}

\begin{figure}
  \centering
    \includegraphics[width=1.0\linewidth] 
  {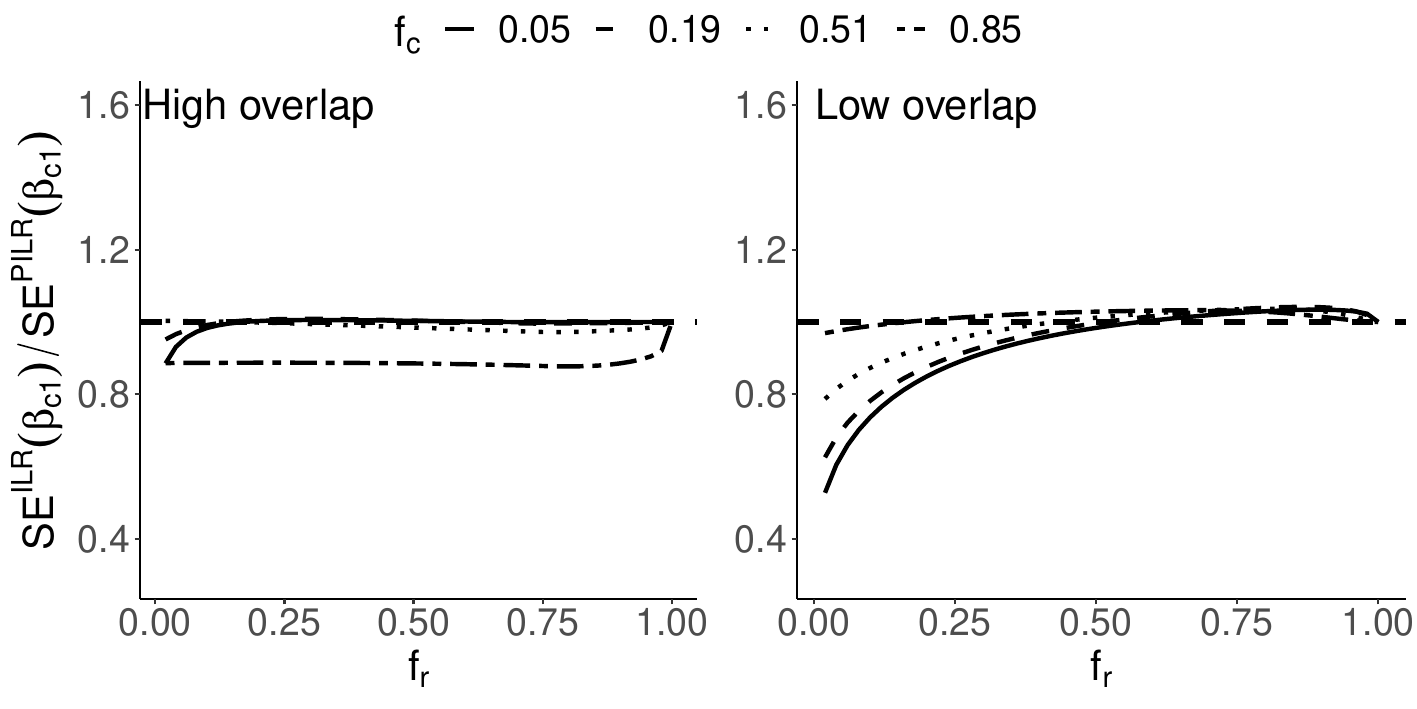}
     \caption{ Relative standard error of ILR and PILR estimates of propensity model parameter $\beta_{c1}$  } \label{fig:counterbetaPILR}
\end{figure}

\begin{figure}
  \centering
    \includegraphics[width = 1.0\linewidth]
  {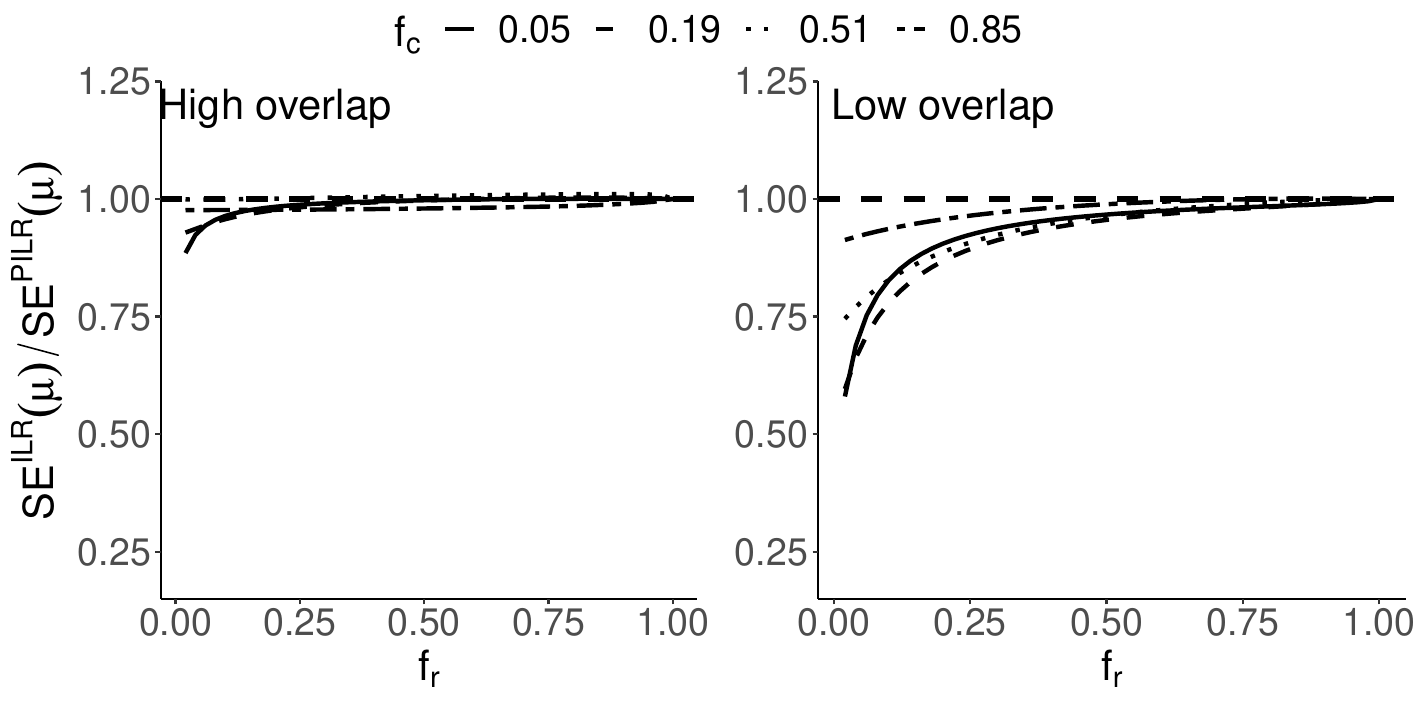}
     \caption{ Relative standard error of ILR and PILR estmates of population mean $\mu$ } \label{fig:countermuPILR}
\end{figure}

We use formulas of Section \ref{sec:asymvars}  to compute variances of estimates ${\mu}$ and ${\boldsymbol{\beta}},$  for each considered method on the grid of sampling fractions $f_r \in \left[ 0.02, 1 \right]$ and values of $f_c = \left\{ 0.05, 0.19,0.51,0.85 \right\}$, for high and low  overlap in variable-defined domains. Figures \ref{fig:counterbeta} and \ref{fig:countermu} show ratios of standard errors of ILR and CLW estimates of $\mu$ and $\beta_{c1}$. Similarly, Figures \ref{fig:counterbetaPILR} and \ref{fig:countermuPILR} show ratios of standard errors of ILR and PILR estimates of $\mu$ and $\beta_{c1}$. 
Notice that design variance $\boldsymbol{D}_m$ enters variance formula \eqref{eq:varmuM} and  \eqref{eq:varbetaM} for both estimates of $\mu$ and $\boldsymbol{\beta}$. This explains similar dependencies of standard error ratios on sampling fractions $f_c$ and $f_r$  in Figures for estimates of $\mu$ and $\beta_{c1}$. 

Figure \ref{fig:counterbeta} demonstrates that CLW is the most efficient (lowest standard error) method for estimating response propensity when sampling fractions $f_r$ approach $1$. 
In this limit CLW pseudo-likelihood \eqref{eq:chen2} practically matches exact Bernoulli likelihood \eqref{eq:chen1} for non-probability sample indicator $I_{ci}$, which explains CLW advantage over the other methods. This advantage is not observed for the estimates of population means in Figure \ref{fig:countermu}. 

Relative efficiency of the CLW method disappears for smaller sampling fraction $f_r$, which is the case for majority of probability surveys used in statistical offices. This effect is particularly pronounced for the low overlap scenario.  The CLW method becomes increasingly inefficient relative to ILR when non-probability fraction $f_c$ is large while probability sampling fraction  $f_r$ is small. This scenario is practically relevant when large administrative data, usually viewed as convenience sample, collects information which can complement and improve estimates from a relatively small probability survey.


\section{Simulation study}\label{sec:sims}
We conduct a Monte Carlo simulation study to compare the relative performances of the CLW, ILR, and PILR estimates under various scenarios for sample sizes, sampling fractions, and degrees of overlap of shared covariates of the convenience and reference samples, that allow us to validate our theoretical results in \eqref{eq:D.CLW}, \eqref{eq:D.ILR} and \eqref{eq:D.PILR}. A secondary purpose of our simulation study is to demonstrate the validity of our theory-derived, plug-in variance estimators of \eqref{eq:varmuM} and \eqref{eq:varbetaM} for variances of the estimated parameters $\mu$ and $\boldsymbol{\beta}$. 
Our simulation study also demonstrates that inference for ILR remains robust no matter the degree of overlap in units between the reference and convenience samples in the case of independent sampling into each. 

\subsection{Simulation setup}

Table \ref{tab:sim_scenarios} outlines the various scenarios utilized in a Monte Carlo (MC) simulation study. In each scenario, convenience and reference samples are drawn $A=1,000$ times from the \emph{same} population. The convenience sample selection probabilities were generated using expression \eqref{eq:numtrueNP} with true coefficient value $\beta_{c1}=1.$ The outcome variable $y_i$ is generated from the normal distribution as $y_i \sim N \left(1 + x_i, 1.5^2 \right)$, where $x_i$ is generated as the independent standard normal variate, so that mean $\mu=1$.

Each of the scenarios is implemented under low and high overlap settings, as described in Section \ref{sec:theorstudy} and illustrated in Figure \ref{fig:overlap}.  Sampling fractions $f_c$ are made equal to approximately $0.10$ or $0.01$ by setting intercept $\beta_{c0}$ in the non-probability survey participation probabilities model \eqref{eq:numtrueNP} to, respectively, $-2.5$ or $-5.0$, while keeping the model coefficient $\beta_{c1}=1.$ Reference samples with fractions $f_r$ equal to $0.01$ and $0.10$ are selected using the probability proportional to size (PPS)  without replacement sampling using \texttt{sampling} R-package by \citet{samplingpkg2021}. The large convenience sample coupled with selecting the whole finite population as the reference sample in Scenario S7 in Table \ref{tab:sim_scenarios} allows us to study the ``ultimate overlap'' case in which all units of the large non-probability sample appear in the stacked set twice. This scenario is implemented by setting intercept $\beta_{c0}=0$.

For this study, reference sample inclusion probabilities $\pi_r$ are assumed to be know for all units in stacked set $S$.

\begin{table}
\centering
 \caption{Summary of considered simulation scenarios.}  \label{tab:sim_scenarios}
\begin{tabular}{crr|rr}
 & \multicolumn{2}{c|}{Sampling} & \multicolumn{2}{c} {Sample} \\  
 & \multicolumn{2}{c|}{fraction} & \multicolumn{2}{c} {size} \\ [3 pt] \hline
Scenario & {$f_c$} & {$f_r$} & {$n_c$} & {$n_r$}\\  \hline
& &  &   \\  [-13 pt]
  & \multicolumn{4}{c}{}\\[-10 pt] 
S1 & 0.01 & 0.01  & 600 & 600  \\
S2 & 0.01 & 0.01  & 100 & 100   \\  
  & \multicolumn{4}{c}{}\\[-10 pt] 
S3 & 0.10 & 0.10  & 600 & 600  \\
S4 & 0.10 & 0.10  & 100 & 100  \\ 
  & \multicolumn{4}{c}{}\\[-10 pt] 
S5 & 0.01 & 0.10  & 100 & 1,000  \\ 
S6 & 0.10 & 0.01  & 1,000 & 100  \\ 
  & \multicolumn{4}{c}{}\\[-10 pt] 
S7 & 0.50 & 1.00  & 500 & 1,000  
\end{tabular} 
\end{table}

\subsection{Simulation results}\label{sec:simsresults}
Table \ref{tab:rmse_summary} reports root mean squared errors (RMSE) for estimators of parameters ${\beta}_{c1}$ and ${\mu}$ for scenarios S1-S6 listed in Table \ref{tab:sim_scenarios}. The top three lines of the table report results for $\hat{\beta}_{c1}$ and the bottom three lines show results for $\hat{\mu}$, for each of the methods. The left half of the table provides results under the high overlap in covariate-defined domains, and the right half of the table shows results under the low overlap.  The RMSEs are computed as 
$\text{RMSE}(\hat{\mu})= \sqrt{{A^{-1}} \sum\nolimits_{\alpha=1}^{A}(\hat{\mu}^{(\alpha)}-\mu)^2 }$ and
$\text{RMSE}(\hat{\beta}_{c1})= \sqrt{{A^{-1}} \sum\nolimits_{\alpha=1}^{A}(\hat{\beta}_{c1}^{(\alpha)}-{\beta_{c1}})^2 }$, and $\hat{\mu}^{(\alpha)}, \hat{\beta}_{c1}^{(\alpha)}$ are estimates of parameters at simulation run $\alpha$ based on respective methods.

The ILR approach provides better results compared to the other methods, while the PILR approach gives better results compared to the CLW for both $\hat{\mu}$ and $\hat{\beta}_{c1}$ in almost every scenario considered. The results are closer for the high overlap cases, and the differences become more pronounced under the low overlap. As expected, for all estimators, RMSEs are generally lower when the sample size is larger. We also note relatively large values of RMSE in the CLW approach under the low overlap cases of scenarios S2, S4, and S6 where the reference sample is small.

\begin{table*}
\centering
\caption{ Root mean square error (RMSE) of estimated parameters ${\beta}_{c1}$ and $\mu$ under high and low overlap for scenarios S1-S6 of Table \ref{tab:sim_scenarios}} 
\label{tab:rmse_summary}
\begin{tabular}{ c | c c c c c c | c c c c c c }
 & \multicolumn{6}{c|}{High overlap} & \multicolumn{6}{c}{Low overlap} \\
 & S1 & S2 & S3 & S4 &  S5 & S6 & S1 & S2 & S3 & S4 &  S5 & S6
 \\  \hline 
  & \multicolumn{12}{c}{} \\ [-10 pt] 

$\hat{\beta}_{ILR}$	   &    0.07	   &    0.17	   &    0.08	   &    0.22	   &    0.12	   &    0.14	   &    0.08	   &    0.22	   &    0.10	   &    0.25	   &    0.14	   &    0.15 \\
$\hat{\beta}_{PILR}$	   &    0.08	   &    0.22	   &    0.08	   &    0.22	   &    0.12	   &    0.16	   &    0.18	   &    0.46	   &    0.13	   &    0.31	   &    0.17	   &    0.26 \\
$\hat{\beta}_{CLW}$	   &    0.09	   &    0.26	   &    0.09	   &    0.29	   &    0.12	   &    0.23	   &    0.22	   &    1.50	   &    0.21	   &    1.22	   &    0.19	   &    0.68 \\

  & \multicolumn{12}{c}{}\\[-10 pt] 

$\hat{\mu}_{ILR}$	   &    0.13	   &    0.33	   &    0.12	   &    0.27	   &    0.30	   &    0.14	   &    0.13	   &    0.32	   &    0.12	   &    0.29	   &    0.32	   &    0.13 \\
$\hat{\mu}_{PILR}$	   &    0.14	   &    0.38	   &    0.12	   &    0.29	   &    0.31	   &    0.17	   &    0.22	   &    0.54	   &    0.15	   &    0.37	   &    0.35	   &    0.25 \\
$\hat{\mu}_{CLW}$	   &    0.15	   &    0.41	   &    0.14	   &    0.33	   &    0.31	   &    0.24	   &    0.26	   &    0.68	   &    0.23	   &    0.64	   &    0.37	   &    0.60 \\

\end{tabular}
\end{table*}


We next examine a more detailed summary for scenarios S5 and S6 of uneven sample sizes between convenience and reference samples in Table \ref{tab:beta_frac_mixed} (for $\beta_{c1}$) and Table \ref{tab:mu_frac_mixed} (for $\mu$). In each MC iteration, we estimated variances by using a plug-in variance estimator based on the asymptotic variance formulas presented in Section \ref{sec:asymvars}. 
Let $\widehat{V}^{(\alpha)}$ denote variance estimate for a parameter ($\mu$ or $\beta_{c1}$) obtained at MC iteration $\alpha$. It can be calculated from expressions \eqref{eq:varmuM} or \eqref{eq:varbetaM} by: 1) plugging-in estimated values of participation probabilities $\hat{\pi}_{ci}$; and 2) replacing  summation over population with weighted sums over non-probability and probability samples. As usual, sampling weights are equal to inverse of estimated participation $\hat{\pi}_{ci}^{-1}$ and selection $\pi_{ri}^{-1}$ probabilities .

\begin{table*}
\centering
\caption{ Estimate of propensity parameter ${\beta}_{c1}$ for sample fractions $\left(f_c,f_r\right) = \left(0.01, 0.1 \right)$  and $\left(0.1, 0.01 \right)$ (scenarios S5 and S6)}
\label{tab:beta_frac_mixed} 
\begin{tabular}{c | c c c c | c c c c}
 
 & \multicolumn{4}{c|}{High overlap} & \multicolumn{4}{c}{Low overlap} \\
 & $Mean$ & $SE$ & $\overline{\widehat{SE}}$ &   $95\%CI$ & 
   $Mean$ & $SE$ & $\overline{\widehat{SE}}$ & $95\%CI$  \\  \hline 
  & \multicolumn{8}{c}{} \\ [-10 pt]
 & \multicolumn{8}{c}{Scenario S5: $n_c=100, n_r=1000$} \\
ILR & 1.00 & 0.12 & 0.11      & 0.95 & 1.01 & 0.14 & 0.13   & 0.95 \\
PILR & 1.00 & 0.12 & 0.12     & 0.95 & 1.03 & 0.16 & 0.15   & 0.94 \\
CLW & 1.01 & 0.12 & 0.12      & 0.95 & 1.05 & 0.19 & 0.16   & 0.93 \\
 & \multicolumn{8}{c}{Scenario S6: $n_c=1000, n_r=100$} \\
ILR & 1.00 & 0.14 & 0.14     &  0.94 & 1.01 & 0.15 & 0.15   & 0.95 \\
PILR & 1.01 & 0.16 & 0.16    &  0.95 & 1.05 & 0.26 & 0.25   & 0.94 \\
CLW & 1.05 & 0.23 & 0.22     &  0.95 & 1.29 & 0.62 & 0.51    & 0.92 \\

\end{tabular}
\end{table*}

\begin{table*}
\centering

\caption{ Estimate of population mean ${\mu}$ for sample fractions $\left(f_c,f_r\right) = \left(0.01, 0.1 \right)$  and $\left(0.1, 0.01 \right)$ (scenarios S5 and S6)}
\label{tab:mu_frac_mixed} 

\begin{tabular}{c | c c  c c | c c  c c}
 
& \multicolumn{4}{c|}{High overlap} & \multicolumn{4}{c}{Low overlap} \\
 & $Mean$ & $SE$ & $\overline{\widehat{SE}}$ &   $95\%CI$ & 
   $Mean$ & $SE$ & $\overline{\widehat{SE}}$ &   $95\%CI$  \\  \hline 
  & \multicolumn{8}{c}{} \\ [-10 pt]
 & \multicolumn{8}{c}{Scenario S5: $n_c=100, n_r=1000$ } \\
ILR & 1.05 & 0.30  & 0.29       & 0.88 & 1.04 & 0.30 & 0.28      & 0.89 \\
PILR & 1.04 & 0.31 & 0.30       & 0.88 & 1.02 & 0.34 & 0.34      & 0.90 \\
CLW & 1.04 & 0.31  & 0.30       & 0.88 & 1.00 & 0.35 & 0.37      & 0.91 \\
 & \multicolumn{8}{c}{Scenario S6: $n_c=1000, n_r=100$} \\
ILR & 1.02 & 0.14  & 0.14       & 0.95 & 1.01 & 0.13 & 0.13      & 0.96 \\
PILR & 1.00 & 0.17 & 0.18       & 0.96 & 0.96 & 0.24 & 0.27      & 0.96 \\
CLW & 0.96 & 0.23  & 0.29       & 0.96 & 0.75 & 0.54 & Inf       & 0.95 \\
\end{tabular}
\end{table*}

The columns in both tables are defined, as follows: 
\begin{itemize}
    \item ``$Mean$'' is the mean estimate of a parameter ($\mu$ or $\beta_{c1}$) over the MC simulations; 
     \item ``$SE$'' is the standard error of the estimates computed over MC simulations, so it is a repeated sampling-derived variance estimator;
     \item ``$\overline{\widehat{SE}}$'' is the average over MC simulations computed based on the plug-in variance estimates $\widehat{V}^{(\alpha)}$  for parameters ($\mu$ or $\beta_{c1}$). 
     We first compute  $\widehat{SE}$ 
     as the square root of the plug-in variance estimate for each MC simulation iteration $\sqrt{\widehat{V}^{(\alpha)}}$. 
     $\widehat{SE}$ is then used to construct $95\%$ confidence intervals and to assess coverage.  
     Finally, we construct $\overline{\widehat{SE}}$ as the average $\overline{\widehat{SE}}=\sqrt{A^{-1}\sum\nolimits_{\alpha=1}^{A}\widehat{V}^{(\alpha)}}$
     over the MC simulation iterations to facilitate the comparison with $SE$.
     \item ``$95\%CI$'' shows the percentage of times, over the MC simulation runs, the true parameter is covered by the $95$ percent normal confidence intervals constructed based on $\sqrt{\widehat{V}^{(\alpha)}}$. 
\end{itemize}
We generally observe relatively high values of $\overline{\widehat{SE}}$ for CLW and to a lesser degree for PILR. So, though the methods may achieve nominal coverage, the resulting average of interval widths over the MC iterations (which are related to $\overline{\widehat{SE}}$) are wider, indicating that CLW and PILR estimators are less efficient than that of ILR.  This relative efficiency difference is particularly pronounced in the low overlap setting.

We note that in the high overlap case in scenario S6, where the reference sample size is relatively small, there is about 7\% bias in the estimates of $\mu$ based on all methods. Respective normal confidence intervals do not provide nominal coverage of the true mean. At the same time, estimates of $\beta_{c1}$ are nearly unbiased and equally efficient for all approaches. 
{In smaller-sized convenience samples nonlinear dependence of an estimator of $\mu$ on  estimated participation probabilities  may result in biased estimates of the mean. Since $\hat{\mu}$ is a consistent estimator, the bias would decrease in larger samples. Bias of estimates for high overlap and small convenience samples also can be seen on boxplots in Figure  \ref{fig:est_mu_0.1-01} of Appendix \ref{sec:SimResultsTblPlots}.}

Lastly, the very similar values of $SE$ and $\overline{\widehat{SE}}$ across the $3$ approaches in Tables  \ref{tab:beta_frac_mixed} and Table \ref{tab:mu_frac_mixed} demonstrate the validity of the theoretically-derived plug-in variance estimator. In some cases in Table \ref{tab:mu_frac_mixed},  estimates of standard error deviate from $SE$ over the MC simulations. We attribute such discrepancies to biased estimates of participation probabilities, i.e. CLW estimate for Scenario S6 under high overlap, or to nonlinear dependence of population mean $\mu$ on $\pi_c$ and small non-probability sample size $n_c$ (ILR estimate for Scenario 5 under low overlap). These cases need to be further studied for better explanation.     

\begin{table}
\centering
\caption{ Estimates of propensity parameter ${\beta}_{c1}$  and population mean ${\mu}$ for  $S_r=U$, non-probability sampling fraction $f_c = 0.5$, population size $N=1,000$ (scenario S7) }
\label{tab:fullpop}
\scalebox{0.95} {
\begin{tabular}{c | c c c  c }    
 & $Mean$ & $SE$ & $\overline{\widehat{SE}}$ &    $95\%CI$ \\ 
 \hline
   &  \multicolumn{4}{c}{}  \\ [-10 pt]
$\hat{\beta}_{ILR/PILR}$ & 1.01 & 0.11 & 0.11        & 0.95  \\
$\hat{\beta}_{CLW}$      & 1.01 & 0.09 & 0.09        & 0.94  \\
  &  \multicolumn{4}{c}{}  \\ [-10 pt]
$\hat{\mu}_{ILR/PILR}$   &  1.00 & 0.07 & 0.07       & 0.94 \\
$\hat{\mu}_{CLW}$        &  1.00 & 0.08 & 0.08       & 0.94 
\end{tabular}
}
\end{table}


Table \ref{tab:fullpop} shows estimates for the ``ultimate overlap'' scenario S7 where the whole finite population is available (and used as the reference sample), so that units of the large convenience sample appear in the stacked set to two samples twice. In this case, sampling fractions and sizes are large and there is also a high overlap in the covariate-defined domains. When the whole population is used as the reference sample, CLW pseudo-likelihood \eqref{eq:chen2} coincides with the population-based likelihood \eqref{eq:chen1}, thus making CLW the optimal method. ILR and PILR are idetical in this case since pseudo-likelihood \eqref{eq:wang3} of the PILR method coincides with the ILR likelihood \eqref{eq:ilr1}. As expected, CLW based estimates of ${\beta}_{c1}$ are more efficient than ILR/PILR based estimates. Interestingly, variances of the estimates of $\mu$ are still smaller for ILR/PILR methods. These observations are also confirmed by the numerical study presented in Figures \ref{fig:counterbeta} and \ref{fig:countermu} of Section \ref{sec:theorstudy}. 

Simulation results verify that variance expressions \eqref{eq:varmuM} and \eqref{eq:varbetaM} work accurately for all considered estimation methods. 

As mentioned in Section \ref{sec:intro}, there is a perception that the existence of the overlapping units in the stacked sample set would prohibit formulation of a likelihood under the independence assumption for the ILR and PILR methods. The closeness of the estimated $\overline{\widehat{SE}}$ and MC simulation based $SE$ numbers verifies formulas \eqref{eq:varmuM} and \eqref{eq:varbetaM} and, indirectly, confirms the validity of the conditional independence assumption.  
At the same time, if overlapping units could be identified, this information could be used to modify the CRISP formula as in \eqref{eq:identityILRwithO} or \eqref{eq:identityILRwithO_Sc}, potentially leading to more efficient estimation.

Additional results for scenarios S1-S4 are reported in Tables \ref{tab:beta_frac_0.01} - \ref{tab:mu_frac_0.1} of Appendix \ref{sec:SimResultsTblPlots}. Boxplots in Figures \ref{fig:est_beta0.01} -  \ref{fig:est_mu_0.1-01} in Appendix \ref{sec:SimResultsTblPlots} show the distribution of relative biases of respective estimated parameters under  scenarios S1-S6.


\section{Concluding remarks}\label{sec:conclusion}

In this paper we reviewed $3$ recent approaches for estimation of a non-probability survey participation probabilities.  We showed that in many practical settings the likelihood based ILR approach provides for more efficient estimates of the model parameters and the IPW estimate of the finite population mean.

The ILR and PILR methods are based on the CRISP formula relating inclusion probabilities of unobserved sample indicators to the probability that can be modeled over an observed stacked samples set. We proved that, conditional on covariates, the indicator variables on the stacked set are nearly independent under the assumption of conditional independence of the samples inclusion indicators. This property allows the use of a likelihood for independent Bernoulli variables formulated over the observed data.  An advantage for the ILR approach that results from this property is that the construction of a valid Bernoulli likelihood on the observed pooled sample allows for application of standard statistical estimation methods. 

We derived asymptotic variances of the estimators and compared theoretical properties of the CLW, ILR, and PILR methods. The formulas revealed that the efficiency of the estimates is impacted by the design-based variance over the repeated probability samples. Small reference samples induce increases in variances of the target parameters in all methods; however, the degree of the effect varies among the methods with the ILR approach being the least susceptible to the effect of a small reference sample size.

The variances are also affected by the degree of the overlap in variable-defined domains between the two samples. Low overlap may cause instability in the estimates even with relatively large reference samples. Any samples combining method depends on the existence of a set of good covariates shared by both samples. Machine learning methods can be useful for the effective modeling. However, ML methods such as classification and regression trees may result in the low-overlap terminal nodes having a small number of units from either sample. The ILR based likelihood could be used as an objective function in these methods, making them more effective.

The methodology is general and can be extended in various ways. At present, for the ILR approach, reference sample inclusion probabilities were assumed to be known for units in both samples. \citet{savitsky2023} co-modeled them along with the non-probability sample participation probabilities and used splines in place of the linear regression model. This approach could be extended to include the outcome variables in the modeling as well. 

\bibliographystyle{imsart-nameyear.bst} 
\bibliography{ref.bib}       

\begin{thebibliography}{23}

\bibitem[\protect\citeauthoryear{Beaumont and Rao}{2021}]{2021Beaumont_Rao}
\begin{barticle}[author]
\bauthor{\bsnm{Beaumont},~\bfnm{Jean-Francois}\binits{J.-F.}} \AND
  \bauthor{\bsnm{Rao},~\bfnm{J.~N.~K.}\binits{J.~N.~K.}}
(\byear{2021}).
\btitle{{Pitfalls of making inferences from non-probability samples: Can data
  integration through probability samples provide remedies}}.
\bjournal{The Survey Statistician}
\bvolume{83}
\bpages{11 -- 22}.
\end{barticle}
\endbibitem

\bibitem[\protect\citeauthoryear{Beresovsky}{2019}]{beresovsky2019}
\begin{bmisc}[author]
\bauthor{\bsnm{Beresovsky},~\bfnm{Vladislav}\binits{V.}}
(\byear{2019}).
\btitle{On application of a response propensity model to estimation from web
  samples}.
\bnote{\href{https://www.researchgate.net/publication/333915871_On_application_of_a_response_propensity_model_to_estimation_from_web_samples}{In
  ResearchGate}}.
\end{bmisc}
\endbibitem

\bibitem[\protect\citeauthoryear{Bethlehem}{2010}]{2010Bethlehem}
\begin{barticle}[author]
\bauthor{\bsnm{Bethlehem},~\bfnm{J.}\binits{J.}}
(\byear{2010}).
\btitle{Selection Bias in Web Surveys}.
\bjournal{International Statistical Review}
\bvolume{78}
\bpages{161 –- 188}.
\end{barticle}
\endbibitem

\bibitem[\protect\citeauthoryear{Chen, Li and Wu}{2020}]{2020_ChenLiWu}
\begin{barticle}[author]
\bauthor{\bsnm{Chen},~\bfnm{Yilin}\binits{Y.}},
  \bauthor{\bsnm{Li},~\bfnm{Pengfei}\binits{P.}} \AND
  \bauthor{\bsnm{Wu},~\bfnm{Changbao}\binits{C.}}
(\byear{2020}).
\btitle{Doubly Robust Inference With Nonprobability Survey Samples}.
\bjournal{Journal of the American Statistical Association}
\bvolume{115}
\bpages{2011-2021}.
\bdoi{10.1080/01621459.2019.1677241}
\end{barticle}
\endbibitem

\bibitem[\protect\citeauthoryear{Elliott}{2009}]{2009elliot}
\begin{barticle}[author]
\bauthor{\bsnm{Elliott},~\bfnm{Michael~R.}\binits{M.~R.}}
(\byear{2009}).
\btitle{Combining Data from Probability and Non-Probability Samples Using
  Pseudo-Weights}.
\bjournal{Survey Practice}
\bvolume{2}
\bpages{813--845}.
\end{barticle}
\endbibitem

\bibitem[\protect\citeauthoryear{Elliott and Valliant}{2017}]{2017elliot}
\begin{barticle}[author]
\bauthor{\bsnm{Elliott},~\bfnm{Michael~R.}\binits{M.~R.}} \AND
  \bauthor{\bsnm{Valliant},~\bfnm{Richard}\binits{R.}}
(\byear{2017}).
\btitle{{Inference for Nonprobability Samples}}.
\bjournal{Statistical Science}
\bvolume{32}
\bpages{249 -- 264}.
\bdoi{10.1214/16-STS598}
\end{barticle}
\endbibitem

\bibitem[\protect\citeauthoryear{Firth}{1993}]{Firth1993}
\begin{barticle}[author]
\bauthor{\bsnm{Firth},~\bfnm{David}\binits{D.}}
(\byear{1993}).
\btitle{Bias reduction of maximum likelihood estimates}.
\bjournal{Biometrika}
\bvolume{80}
\bpages{27--38}.
\end{barticle}
\endbibitem

\bibitem[\protect\citeauthoryear{Heinze and
  Schemper}{2002}]{HeinzeSchemper2002}
\begin{barticle}[author]
\bauthor{\bsnm{Heinze},~\bfnm{G.}\binits{G.}} \AND
  \bauthor{\bsnm{Schemper},~\bfnm{M.}\binits{M.}}
(\byear{2002}).
\btitle{A solution to the problem of separation in logistic regression}.
\bjournal{Statistics in Medicine}
\bvolume{21}
\bpages{2409--2419}.
\end{barticle}
\endbibitem

\bibitem[\protect\citeauthoryear{Kim and Tam}{2021}]{2021KimTam}
\begin{barticle}[author]
\bauthor{\bsnm{Kim},~\bfnm{J.~K.}\binits{J.~K.}} \AND
  \bauthor{\bsnm{Tam},~\bfnm{S.~M.}\binits{S.~M.}}
(\byear{2021}).
\btitle{Data Integration by Combining Big Data and Survey Sample Data for
  Finite Population Inference}.
\bjournal{International Statistical Review}
\bvolume{89}
\bpages{382–-401}.
\end{barticle}
\endbibitem

\bibitem[\protect\citeauthoryear{Kruskal and Mosteller}{1980}]{Kruskal1980}
\begin{barticle}[author]
\bauthor{\bsnm{Kruskal},~\bfnm{William~H.}\binits{W.~H.}} \AND
  \bauthor{\bsnm{Mosteller},~\bfnm{Frederick}\binits{F.}}
(\byear{1980}).
\btitle{Representative Sampling, IV: The History of the Concept in Statistics,
  1895-1939}.
\bjournal{International Statistical Review}
\bvolume{48}
\bpages{169--195}.
\end{barticle}
\endbibitem

\bibitem[\protect\citeauthoryear{Liu, Scholtus and
  De~Waal}{2023}]{2022LiuScholtusWaal}
\begin{barticle}[author]
\bauthor{\bsnm{Liu},~\bfnm{An-Chiao}\binits{A.-C.}},
  \bauthor{\bsnm{Scholtus},~\bfnm{Sander}\binits{S.}} \AND
  \bauthor{\bsnm{De~Waal},~\bfnm{Ton}\binits{T.}}
(\byear{2023}).
\btitle{{Correcting Selection bias in Big Data by pseudo-weighting}}.
\bjournal{Journal of Survey Statistics and Methodology}
\bvolume{11}
\bpages{1181–-1203}.
\end{barticle}
\endbibitem

\bibitem[\protect\citeauthoryear{Meng}{2018}]{2018Meng}
\begin{barticle}[author]
\bauthor{\bsnm{Meng},~\bfnm{X.~L.}\binits{X.~L.}}
(\byear{2018}).
\btitle{Statistical Paradises and Paradoxes in Big Data (I): Law of Large
  Populations, Big Data Paradox, and the 2016 US Presidential Election}.
\bjournal{Annals of Applied Statistics}
\bvolume{12}
\bpages{685 -- 726}.
\end{barticle}
\endbibitem

\bibitem[\protect\citeauthoryear{Neyman}{1934}]{neyman1934}
\begin{barticle}[author]
\bauthor{\bsnm{Neyman},~\bfnm{Jerzy}\binits{J.}}
(\byear{1934}).
\btitle{On the Two Different Aspects of the Representative Method: The Method
  of Stratified Sampling and the Method of Purposive Selection}.
\bjournal{Journal of the Royal Statistical Society}
\bvolume{97}
\bpages{558--625}.
\end{barticle}
\endbibitem

\bibitem[\protect\citeauthoryear{Rao}{2021}]{2020Rao}
\begin{barticle}[author]
\bauthor{\bsnm{Rao},~\bfnm{J.~N.~K.}\binits{J.~N.~K.}}
(\byear{2021}).
\btitle{{On Making Valid Inferences by Integrating Data from Surveys and Other
  Sources}}.
\bjournal{Sankhya}
\bvolume{83}
\bpages{242 -- 272}.
\end{barticle}
\endbibitem

\bibitem[\protect\citeauthoryear{Rosenbaum and
  Rubin}{1983}]{RosenbaumRubin1983}
\begin{barticle}[author]
\bauthor{\bsnm{Rosenbaum},~\bfnm{Paul~R.}\binits{P.~R.}} \AND
  \bauthor{\bsnm{Rubin},~\bfnm{Donald~B.}\binits{D.~B.}}
(\byear{1983}).
\btitle{{The central role of the propensity score in observational studies for
  causal effects}}.
\bjournal{Biometrika}
\bvolume{70}
\bpages{41-55}.
\bdoi{10.1093/biomet/70.1.41}
\end{barticle}
\endbibitem

\bibitem[\protect\citeauthoryear{Rubin}{1976}]{Rubin1976}
\begin{barticle}[author]
\bauthor{\bsnm{Rubin},~\bfnm{Donald~B.}\binits{D.~B.}}
(\byear{1976}).
\btitle{{Inference and missing data}}.
\bjournal{Biometrika}
\bvolume{63}
\bpages{581-592}.
\bdoi{10.1093/biomet/63.3.581}
\end{barticle}
\endbibitem

\bibitem[\protect\citeauthoryear{Savitsky et~al.}{2023}]{savitsky2023}
\begin{barticle}[author]
\bauthor{\bsnm{Savitsky},~\bfnm{Terrance~D.}\binits{T.~D.}},
  \bauthor{\bsnm{Williams},~\bfnm{Matthew~R.}\binits{M.~R.}},
  \bauthor{\bsnm{Gershunskaya},~\bfnm{Julie}\binits{J.}} \AND
  \bauthor{\bsnm{Beresovsky},~\bfnm{Vladislav}\binits{V.}}
(\byear{2023}).
\btitle{Methods for combining probability and nonprobability samples under
  unknown overlaps}.
\bjournal{Statistics in Transition New Series}
\bvolume{24}
\bpages{1 —- 34}.
\end{barticle}
\endbibitem

\bibitem[\protect\citeauthoryear{Tille and Matei}{2021}]{samplingpkg2021}
\begin{barticle}[author]
\bauthor{\bsnm{Tille},~\bfnm{Yves}\binits{Y.}} \AND
  \bauthor{\bsnm{Matei},~\bfnm{Alina}\binits{A.}}
(\byear{2021}).
\btitle{\texttt{sampling}: Survey Sampling}.
\bnote{\href{ https://CRAN.R-project.org/package=sampling}{
  https://CRAN.R-project.org/package=sampling}}.
\end{barticle}
\endbibitem

\bibitem[\protect\citeauthoryear{Valliant}{2020}]{2020Valliant}
\begin{barticle}[author]
\bauthor{\bsnm{Valliant},~\bfnm{Richard}\binits{R.}}
(\byear{2020}).
\btitle{{Comparing Alternatives for Estimation from Nonprobability Samples}}.
\bjournal{Journal of Survey Statistics and Methodology}
\bvolume{8}
\bpages{231 -- 263}.
\end{barticle}
\endbibitem

\bibitem[\protect\citeauthoryear{VanderWeele and
  Shpitser}{2011}]{2011VanderWeele_Shpitser}
\begin{barticle}[author]
\bauthor{\bsnm{VanderWeele},~\bfnm{T.~J.}\binits{T.~J.}} \AND
  \bauthor{\bsnm{Shpitser},~\bfnm{I.}\binits{I.}}
(\byear{2011}).
\btitle{A new criterion for confounder selection}.
\bjournal{Biometrics}
\bvolume{67}
\bpages{1406 —- 1413}.
\end{barticle}
\endbibitem

\bibitem[\protect\citeauthoryear{Wang, Valliant and Li}{2021}]{2021valliant}
\begin{barticle}[author]
\bauthor{\bsnm{Wang},~\bfnm{L.}\binits{L.}},
  \bauthor{\bsnm{Valliant},~\bfnm{R.}\binits{R.}} \AND
  \bauthor{\bsnm{Li},~\bfnm{Y.}\binits{Y.}}
(\byear{2021}).
\btitle{Adjusted logistic propensity weighting methods for population inference
  using nonprobability volunteer-based epidemiologic cohorts}.
\bjournal{Stat Med.}
\bvolume{40}
\bpages{5237--5250}.
\bdoi{10.1002/sim.9122}
\end{barticle}
\endbibitem

\bibitem[\protect\citeauthoryear{Williams and
  Brick}{2017}]{10.1093/jssam/smx019}
\begin{barticle}[author]
\bauthor{\bsnm{Williams},~\bfnm{Douglas}\binits{D.}} \AND
  \bauthor{\bsnm{Brick},~\bfnm{J~Michael}\binits{J.~M.}}
(\byear{2017}).
\btitle{{Trends in U.S. Face-To-Face Household Survey Nonresponse and Level of
  Effort}}.
\bjournal{Journal of Survey Statistics and Methodology}
\bvolume{6}
\bpages{186-211}.
\bdoi{10.1093/jssam/smx019}
\end{barticle}
\endbibitem

\bibitem[\protect\citeauthoryear{Wu}{2022}]{2022Wu}
\begin{barticle}[author]
\bauthor{\bsnm{Wu},~\bfnm{Changbao}\binits{C.}}
(\byear{2022}).
\btitle{A new criterion for confounder selection}.
\bjournal{Survey Methodology}
\bvolume{48}
\bpages{283 —- 311}.
\end{barticle}
\endbibitem

\end{thebibliography}


\begin{appendix}
\section{Proof of Theorem \ref{theo:2_2}} \label{sec:ProofT2.2}

Our goal is to find $P\left\{ i,j \in S_c, \mid i,j \in S, i,j \in 2U\right\}$:
\begin{align}\label{eq:Sc_Sij}
\begin{split}
&P\left\{ i,j \in S_c, \mid i,j \in S, i,j \in 2U\right\}=\frac{P\left\{ i,j \in S_c \mid i,j \in 2U\right\}}{P\left\{ i,j \in S \mid i,j \in 2U\right\}}.\\
\end{split}
\end{align}
We can view the ``double population'' $2U$ as consisting of two strata, $U_c$ and $U_r$, each of size $N$. We draw stratified sample $S$ from strata $U_c$ and $U_r$ using Poisson sampling with probabilities $\pi_{ci}$ and $\pi_{ri}$, respectively; $i=1,\dots N$.\\
Consider joint probabilities of units $i$ and $j$ to be included into $S_c$:
\begin{eqnarray} \label{eq:S_cij}
 &&  P\left\{ i,j\in S_c \mid i,j \in 2U\right\} \\
 && =P\left\{ i,j\in S_c \mid i,j \in U_c\right\}P\left\{i,j \in U_c \mid i,j \in 2U\right\} \nonumber \\
&& =\pi_{ci}\pi_{cj}\frac{N(N-1)}{2N(2N-1)}, \nonumber
\end{eqnarray}
Similarly, joint probabilities of units $i$ and $j$ to be included into $S_r$ are 
\begin{eqnarray}
&& P\left\{ i,j\in S_r \mid i,j \in 2U\right\} \\
&& =P\left\{ i,j\in S_r \mid i,j \in U_r\right\} P\left\{i,j \in U_r \mid i,j \in 2U\right\} \nonumber\\
&& =\pi_{ri}\pi_{rj}\frac{N(N-1)}{2N(2N-1)}. \nonumber
\end{eqnarray}
The joint probabilities for units $i$ and $j$ selected from different ``strata'' are
\begin{eqnarray}
&& P\left\{ i\in S_c, j \in S_r \mid i,j \in 2U\right\}\\
&& =P\left\{ i\in S_c,j\in S_r  \mid i \in U_c, j \in U_r\right\} \times \nonumber \\
&& P\left\{ i \in U_c, j \in U_r \mid i,j \in 2U\right\} \nonumber\\
&& =\pi_{ci}\pi_{rj}\frac{N^2}{2N(2N-1)}. \nonumber
\end{eqnarray}
Similarly,
\begin{eqnarray}
&& P\left\{ i\in S_r, j \in S_c \mid i,j \in 2U\right\}\\
&& =\pi_{ri}\pi_{cj}\frac{N^2}{2N(2N-1)}. \nonumber
\end{eqnarray}
The joint probability of selection into sample $S$ is the sum of the above parts: 
\begin{eqnarray} \label{eq:Sij}
&& P\left\{ i,j \in S \mid i,j \in 2U\right\}\\
&& =\frac{N(N-1)}{2N(2N-1)}(\pi_{ci}+\pi_{ri})(\pi_{cj}+\pi_{rj}) \nonumber\\
&& +\frac{N}{2N(2N-1)}(\pi_{ri}\pi_{cj}+\pi_{ci}\pi_{rj}). \nonumber
\end{eqnarray}
Using (\ref{eq:S_cij}) and (\ref{eq:Sij}) in (\ref{eq:Sc_Sij}):
\begin{eqnarray}\label{eq:Jointzij}
&& P\left\{ i,j \in S_c, \mid i,j \in S, i,j \in 2U\right\}\\
&&=\frac{\pi_{ci}\pi_{cj}}{(\pi_{ci}+\pi_{ri})(\pi_{cj}+\pi_{rj})+\frac{1}{N-1}(\pi_{ri}\pi_{cj}+\pi_{ci}\pi_{rj})} \nonumber \\
&& =\pi_{zi}\pi_{zj}F, \nonumber
\end{eqnarray}
where
\begin{eqnarray}
&& F=\frac{1}{1+\frac{1}{N-1}{[(1-\pi_{zi})\pi_{zj}+\pi_{zi}(1-\pi_{zj})]}}\\
&& =1-O(\frac{1}{N-1}). \nonumber
\end{eqnarray}
Similarly,
\begin{align}\label{eq:Cov1-zij}
\begin{split}
&P\left\{ i \in S_c, j \in S_r \mid i,j \in S, i,j \in 2U\right\}\\
&=\frac{N}{N-1}\frac{\pi_{ci}\pi_{rj}}{(\pi_{ci}+\pi_{ri})(\pi_{cj}+\pi_{rj})+\frac{1}{N-1}(\pi_{ri}\pi_{cj}+\pi_{ci}\pi_{rj})}\\
&=\frac{N}{N-1}\pi_{zi}(1-\pi_{zj})F.
\end{split}
\end{align}
The covariance is:
\begin{align}\label{eq:Covzij}
\begin{split}
&Cov(I_{zi},I_{zj})=P\left\{ i,j \in S_c, \mid i,j \in S, i,j \in 2U\right\}\\
&-P\left\{ i \in S_c, \mid i,j \in S, i,j \in 2U\right\}P\left\{ j \in S_c, \mid i,j \in S, i,j \in 2U\right\}.
\end{split}
\end{align}
where, by the definition of conditional probability (dropping $i,j \in 2U$, to simplify notation): 
\begin{align}\label{eq:Indzij}
\begin{split}
&P\left\{ i \in S_c, \mid i,j \in S\right\}=\frac{P\left\{ i \in S_c, j \in S \right\}}{P\left\{ i,j \in S \right\}}\\
&=\frac{P\left\{ i, j \in S_c \right\}+P\left\{ i \in S_c, j \in S_r \right\}}{P\left\{ i,j \in S \right\}},
\end{split}
\end{align}
\begin{align}\label{eq:Indzji}
\begin{split}
&P\left\{ j \in S_c, \mid i,j \in S\right\}=\frac{P\left\{ j \in S_c, i \in S \right\}}{P\left\{ i,j \in S \right\}}\\
&=\frac{P\left\{ i , j \in S_c \right\}+P\left\{ i \in S_r, j \in S_c \right\}}{P\left\{ i,j \in S \right\}}.
\end{split}
\end{align}
Hence, putting everything together,
\begin{eqnarray} \label{eq:Covzz}
&& Cov(I_{zi},I_{zj})=\frac{P\left\{ i,j \in S_c\right\}}{P\left\{ i,j \in S \right\}}\\
&& -\frac{P\left\{ i , j \in S_c \right\}+P\left\{ i \in S_c, j \in S_r \right\}}{P\left\{ i,j \in S \right\}}\times \nonumber \\
&& \frac{P\left\{ i , j \in S_c \right\}+P\left\{ i \in S_r, j \in S_c \right\}}{P\left\{ i,j \in S \right\}} \nonumber \\
&& =\pi_{zi}\pi_{zj}F \nonumber\\
&& -(\pi_{zi}\pi_{zj}F+\frac{N}{N-1}\pi_{zi}(1-\pi_{zj})F) \times \nonumber\\
&& (\pi_{zi}\pi_{zj}F+\frac{N}{N-1}\pi_{zj}(1-\pi_{zi})F) \nonumber \\
&& =\pi_{zi}\pi_{zj}F(1-FG), \nonumber
\end{eqnarray}
where $G=(1+\frac{1-\pi_{zj}}{N-1})(1+\frac{1-\pi_{zi}}{N-1})$.

Note that 
\begin{eqnarray} 
&&FG=F(1+\frac{1-\pi_{zj}}{N-1})(1+\frac{1-\pi_{zi}}{N-1})\\
&&=\frac{1+\frac{1}{N-1}[(1-\pi_{zi})+(1-\pi_{zj})]}{1+\frac{1}{N-1}{[(1-\pi_{zi})\pi_{zj}+\pi_{zi}(1-\pi_{zj})]}}\nonumber\\
&&+\frac{\frac{1}{(N-1)^2}(1-\pi_{zi})(1-\pi_{zj})}{1+\frac{1}{N-1}{[(1-\pi_{zi})\pi_{zj}+\pi_{zi}(1-\pi_{zj})]}}\nonumber\\
&&=1+O\left(\frac{1}{N-1}\right).\nonumber
\end{eqnarray}
Thus,
\begin{eqnarray} 
&& Cov(I_{zi},I_{zj})\\
&& =-\pi_{zi}\pi_{zj}\left[1-O\left(\frac{1}{N-1}\right)\right]O\left(\frac{1}{N-1}\right) \nonumber \\
&& =-O\left(\frac{1}{N}\right). \nonumber
\end{eqnarray}
\\

\section{Proof of Theorem \ref{theo:consist}} \label{sec:ProofT5.1}

Consider vector of parameters $\boldsymbol{\eta^T}=(\mu,\boldsymbol{\beta^T})$. The full set of estimating equations is
\begin{equation}\label{eq:T_EE}
\boldsymbol{\Phi}(\boldsymbol{\eta})=
\begin{pmatrix}
    U(\mu)\\
    S(\boldsymbol{\beta})
\end{pmatrix}
=\boldsymbol{0},
\end{equation}
where $U(\mu)=\mathop{\sum}_{i \in U}I_{ci}\pi_{ci}^{-1}(y_i-\mu)$ is a score function for $\mu$  and  $S(\boldsymbol{\beta})$ are score function for respective models for $\beta.$
Let $\boldsymbol{\eta}_0^T=(\mu_0,\boldsymbol{\beta}_0^T)$ be a vector of true parameters satisfying $E\boldsymbol{\Phi}(\boldsymbol{\eta}_0)=\boldsymbol{0}$ and $\hat{\boldsymbol{\eta}}$ a solution of (\ref{eq:T_EE}). From first-order Taylor expansion of $\boldsymbol{\Phi}(\hat{\boldsymbol{\eta}})$ around $\boldsymbol{\eta}_0$, we have
\begin{align}\label{eq:T_Taylor}
\hat{\boldsymbol{\eta}} - \boldsymbol{\eta}_0\doteq\left[E\{\boldsymbol{\phi}(\boldsymbol{\eta}_0)\}\right]^{-1}\boldsymbol{\Phi}(\boldsymbol{\eta}_0),
\end{align}
where 
\begin{equation}\label{eq:T_phimatrix}
E\{\boldsymbol{\phi}(\boldsymbol{\eta})\}=E\frac{\partial \boldsymbol{\Phi}(\boldsymbol{\eta})}{\partial \boldsymbol{\eta}}=
\begin{pmatrix}
    U_{\mu} & U_{\boldsymbol{\beta}} \\
    \boldsymbol{0} & S_{\boldsymbol{\beta}}
\end{pmatrix}
,
\end{equation}
\begin{align*}
U_{\mu} &= E\frac{\partial U}{\partial \mu}  = -\mathop{\sum}_{i \in U}EI_{ci}\pi_{ci}^{-1}=-N, \\
U_{\boldsymbol{\beta}} &= E\frac{\partial U}{\partial \boldsymbol{\beta^T}}  = \mathop{\sum}_{i \in U}(y_i-\mu)EI_{ci}\frac{\partial }{\partial \boldsymbol{\beta^T}}\pi_{ci}^{-1} \\ 
& =-\mathop{\sum}_{i \in U}(y_i-\mu)(1-\pi_{ci})\mathbf{x}_i^T. 
\end{align*}
From (\ref{eq:T_Taylor}), the approximate variance is
\begin{align}\label{eq:T_sandwich}
Var(\hat{\boldsymbol{\eta}}) \doteq \left[E\{\boldsymbol{\phi}(\boldsymbol{\eta}_0)\}\right]^{-1}Var\{\boldsymbol{\Phi}(\boldsymbol{\eta}_0)\}\left[E\{\boldsymbol{\phi}(\boldsymbol{\eta}_0)\}\right]^{-1},
\end{align}
where the inverse of matrix (\ref{eq:T_phimatrix}) comes out to
\begin{equation}\label{eq:T_invphimatrix}
\left[E\{\boldsymbol{\phi}(\boldsymbol{\eta})\}\right]^{-1}=N^{-1}
\begin{pmatrix}
    -1 & \boldsymbol{b}^T \\
    \boldsymbol{0} & NS_{\boldsymbol{\beta}}^{-1}
\end{pmatrix}
,
\end{equation}
where $\boldsymbol{b}=S_{\boldsymbol{\beta}}^{-1}U_{\boldsymbol{\beta}}^T$ and
\begin{align}\label{eq:T_varphi}
Var[\boldsymbol{\Phi}(\boldsymbol{\eta}_0)]=
\begin{pmatrix}
    Var[U(\mu)] & \boldsymbol{C}^T \\
    \boldsymbol{C} & Var[S(\boldsymbol{\beta})]
\end{pmatrix}
,
\end{align}
where $\boldsymbol{C}=Cov[U(\mu),S(\boldsymbol{\beta})].$

Finally, putting (\ref{eq:T_invphimatrix}) and (\ref{eq:T_varphi}) into (\ref{eq:T_sandwich}), we find that


\begin{align}\label{eq:T_final}
& Var(\hat{\boldsymbol{\eta}}) \doteq \\
& \begin{pmatrix}
    N^{-2}\{Var[U(\mu)]-2\boldsymbol{b}^T\boldsymbol{C}+\boldsymbol{b}^TVar[S(\boldsymbol{\beta})]\boldsymbol{b}\} & \;0 \\
    -N^{-1}S_{\boldsymbol{\beta}}^{-1}[\boldsymbol{C}-Var[S(\boldsymbol{\beta})]\boldsymbol{b}] & \;0
\end{pmatrix} \nonumber \\
&+ \begin{pmatrix}
    0 \; & -N^{-1}[\boldsymbol{C^T}-\boldsymbol{b}^TVar[S(\boldsymbol{\beta})]S_{\boldsymbol{\beta}}^{-1} \\
    0 \; & S_{\boldsymbol{\beta}}^{-1}Var[S(\boldsymbol{\beta})]S_{\boldsymbol{\beta}}^{-1}
\end{pmatrix} \nonumber 
.
\end{align}
where the  design variance of $U(\mu)$ under the Poisson sampling is
\begin{align*}\label{eq:T1_components_var1}
Var\{U(\mu)\}&=\mathop{\sum}_{i \in U}Var(I_{ci})\left(\frac{y_i-\mu}{\pi_{ci}}\right)^2 \\
&=\mathop{\sum}_{i \in U}(1-\pi_{ci})\pi_{ci}^{-1}(y_i-\mu)^2.  \nonumber
\end{align*}

If we present $S(\boldsymbol{\beta})$ as a sum of two independent parts $S(\boldsymbol{\beta})=S_c(\boldsymbol{\beta})+S_r(\boldsymbol{\beta}),$
the variance can be written as
\begin{align*}
Var[S(\boldsymbol{\beta})]&=Var[S_c(\boldsymbol{\beta})]+Var[S_r(\boldsymbol{\beta})]. 
\end{align*}
Let us now spell out (\ref{eq:T_final}) for the ILR and PILR methods.

\textbf{The ILR case:}

The score function is
\begin{align*}
S(\boldsymbol{\beta})&=\mathop{\sum}_{i \in S_c}(1-\pi_{zi})(1-\pi_{ci})\mathbf{x}_i-\mathop{\sum}_{i \in S_r}\pi_{zi}(1-\pi_{ci})\mathbf{x}_i\\
&=\mathop{\sum}_{i \in U}[I_{ci}(1-\pi_{zi})-I_{ri}\pi_{zi}](1-\pi_{ci})\mathbf{x}_i,
\end{align*}
that can be decomposed into:
\begin{align*}
S_c(\boldsymbol{\beta})&=\mathop{\sum}_{i \in U}I_{ci}(1-\pi_{zi})(1-\pi_{ci})\mathbf{x}_i,\\
S_r(\boldsymbol{\beta})&=-\mathop{\sum}_{i \in U}I_{ri}\pi_{zi}(1-\pi_{ci})\mathbf{x}_i.
\end{align*}
Next,
\begin{align*}
\boldsymbol{H}_{ILR} &=-S_{\boldsymbol{\beta}} = -E\frac{\partial S(\boldsymbol{\beta})}{\partial \boldsymbol{\beta}} \\
& =\sum_{i \in U}(\pi_{ci}+\pi_{ri})\pi_{zi}(1-\pi_{zi})(1-\pi_{ci})^2\mathbf{x}_i\mathbf{x}_i^T. 
\end{align*}
The variance of $S_c(\boldsymbol{\beta})$ is
\begin{align*}
& \boldsymbol{A}_{ILR} = Var[S_c(\boldsymbol{\beta})] \\ 
&=\mathop{\sum}_{i \in U}\pi_{ci}(1-\pi_{ci})(1-\pi_{zi})^2(1-\pi_{ci})^2\mathbf{x}_i\mathbf{x}_i^T \nonumber \\
&=\boldsymbol{H}_{ILR}-\sum_{i \in U}(\pi_{ri}+1)\pi_{ci}\pi_{zi}(1-\pi_{zi})(1-\pi_{ci})^2\mathbf{x}_i\mathbf{x}_i^T 
\end{align*}
The variance of $S_r(\boldsymbol{\beta})$ is $\boldsymbol{D}_{ILR}=Var_d[\mathop{\sum}_{i \in U}I_{ri}\pi_{zi}(1-\pi_{ci})\mathbf{x}_i]$ is the variance-covariance matrix under the probability sample design.

For the covariance, using the independence of indicators $I_{ci}$ and $I_{ri}$, we have
\begin{align*}
& \boldsymbol{C}_{ILR}  = Cov[U(\mu),S(\boldsymbol{\beta})] \\ 
&= \mathop{\sum}_{i \in U}Cov(I_{ci},I_{ci})\pi_{ci}^{-1}(y_i-\mu)(1-\pi_{zi})(1-\pi_{ci})\mathbf{x}_i \nonumber \\
&= \mathop{\sum}_{i \in U}(y_i-\mu)(1-\pi_{zi})(1-\pi_{ci})^2\mathbf{x}_i. 
\end{align*}

Variance formulas for the parameters in the ILR method follow by putting together the above formulas.

It is easy to see, using regularity condition C2, that 
\begin{align*}
Var({U\mu})&= N^2O\left(\frac{1}{n_c}\right),\\
\boldsymbol{b}_{ILR}^T\boldsymbol{C}_{ILR}&=N^2O\left(\frac{n_c+n_r}{n_cn_r}\frac{n_r}{n_c+n_r}\right),
\end{align*}
\begin{align*}
& \boldsymbol{b}_{ILR}^T(\boldsymbol{A}_{ILR}+\boldsymbol{D}_{ILR})\boldsymbol{b}_{ILR} \\
&=N^2O\left(\frac{(n_c+n_r)^2}{n_c^2n_r^2}\left\{\frac{n_cn_r^2}{(n_c+n_r)^2}+\frac{n_c^2n_r}{(n_c+n_r)^2}\right\}\right),\\ 
&\boldsymbol{H}_{ILR}^{-1}\left(\boldsymbol{A}_{ILR}+\boldsymbol{D}_{ILR}\right)\boldsymbol{H}_{ILR} \\
&=O\left(\frac{(n_c+n_r)^2}{n_c^2n_r^2}\left\{\frac{n_cn_r^2}{(n_c+n_r)^2}+\frac{n_c^2n_r}{(n_c+n_r)^2}\right\}\right).
\end{align*}

After putting it into the variance formulas, we have
\begin{align*}
  Var(\hat{\mu}_{ILR})& =  O \Big( \frac{1}{n_c}+\frac{1}{n_r} \Big), \\
  Var(\boldsymbol{\hat{\beta}}_{ILR})& =O\Big( \frac{1}{n_c}+\frac{1}{n_r}\Big).
\end{align*}
Hence, $\hat{\mu}_{ILR}$ and $\boldsymbol{\hat{\beta}}_{ILR}$ are consistent estimates of, respectively, $\mu$ and $\boldsymbol{\beta}$ and $\hat{\mu}_{ILR}-\mu=O_p(\text{min}(n_c,n_r)^{-1/2})$ and $\boldsymbol{\hat{\beta}}_{ILR}-\boldsymbol{{\beta}}=O_p(\text{min}(n_c,n_r)^{-1/2})$.

\textbf{The PILR case:}

The score function is
\begin{align*} 
    \begin{split}
    S(\boldsymbol{\beta})&=\mathop{\sum}_{i \in S_c}(1-\pi_{{\delta}i})(1-\pi_{ci})\mathbf{x}_i  \\
    &- \mathop{\sum}_{i \in S_r}w_{ri}\pi_{{\delta}i}(1-\pi_{ci})\mathbf{x}_i \\
    \end{split}\\
    &=\mathop{\sum}_{i \in U}\left\{I_{ci}(1-\pi_{{\delta}i})- I_{ri}w_{ri}\pi_{{\delta}i}\right\}(1-\pi_{ci})\mathbf{x}_i,
\end{align*}
that can be decomposed into:
\begin{align*}
    S_c({\boldsymbol{\beta}})&=\mathop{\sum}_{i \in U}(I_{ci}(1-\pi_{{\delta}i})-\pi_{{\delta}i})(1-\pi_{ci})\mathbf{x}_i, \\
    S_r({\boldsymbol{\beta}})&=\mathop{\sum}_{i \in U}(1-I_{ri}w_{ri})\pi_{{\delta}i}(1-\pi_{ci})\mathbf{x}_i. 
\end{align*}
Next,
  \begin{align*} 
\boldsymbol{H}_{PILR} &=-S_{\boldsymbol{\beta}} = -E\frac{\partial S(\boldsymbol{\beta})}{\partial \boldsymbol{\beta}} \\ 
& =\sum_{i \in U}(\pi_{ci}+1)\pi_{{\delta}i}(1-\pi_{{\delta}i})(1-\pi_{ci})^2\mathbf{x}_i\mathbf{x}_i^T. \nonumber 
\end{align*}
The variance of $S_c(\boldsymbol{\beta})$ is
\begin{align*}
\boldsymbol{A}_{PILR} &= Var[S_c(\boldsymbol{\beta})]\\
&=\mathop{\sum}_{i \in U}\pi_{ci}(1-\pi_{{\delta}i})(1-\pi_{ci})^2\mathbf{x}_i\mathbf{x}_i^T\nonumber \\
&-2\mathop{\sum}_{i \in U}\pi_{ci}\pi_{{\delta}i}(1-\pi_{{\delta}i})(1-\pi_{ci})^2\mathbf{x}_i\mathbf{x}_i^T  \nonumber\\
&=\boldsymbol{H}_{PILR}-2\mathop{\sum}_{i \in U}\pi_{ci}\pi_{{\delta}i}(1-\pi_{{\delta}i})(1-\pi_{ci})^2\mathbf{x}_i\mathbf{x}_i^T 
\end{align*}
The variance of $S_r(\boldsymbol{\beta})$ is 
\\ $\boldsymbol{D}_{PILR}=Var_d[\mathop{\sum}_{i \in U}I_{ri}w_{ri}\pi_{{\delta}i}(1-\pi_{ci})\mathbf{x}_i]$,  the variance-covariance matrix under the probability sample design.

For the covariance, using the independence of indicators $I_{ci}$ and $I_{ri}$, we have
\begin{align*}
& \boldsymbol{C}_{PILR}  =Cov[U(\mu),S(\boldsymbol{\beta})]  \\ 
 & = \mathop{\sum}_{i \in U}Cov(I_{ci},I_{ci})\pi_{ci}^{-1}(y_i-\mu)(1-\pi_{{\delta}i})(1-\pi_{ci})\mathbf{x}_i \nonumber \\
&= \mathop{\sum}_{i \in U}(y_i-\mu)(1-\pi_{{\delta}i})(1-\pi_{ci})^2\mathbf{x}_i.
\end{align*}

Variance formulas for the parameters in the PILR method follow by putting together the above formulas.

We can see, using regularity condition C2, that  
\begin{align*}
Var({U\mu})&= N^2O\left(\frac{1}{n_c}\right),\\
\boldsymbol{b}_{PILR}^T\boldsymbol{C}_{PILR}&=N^2O\left({\frac{1}{n_c}}\right),\\
\boldsymbol{b}_{PILR}^T(\boldsymbol{A}_{PILR}+\boldsymbol{D}_{PILR})\boldsymbol{b}_{PILR}&=N^2O\Big[({\frac{1}{n_c^2}})(n_c+\frac{n_c^2}{n_r})\Big],\\ 
\boldsymbol{H}_{PILR}^{-1}(\boldsymbol{A}_{PILR}+\boldsymbol{D}_{PILR})\boldsymbol{H}_{PILR}&=O\Big[\frac{1}{n_c^2}(n_c+\frac{n_c^{2}}{n_r})\Big].
\end{align*}
After putting it into the variance formulas, we have
\begin{align*}
  Var(\hat{\mu}_{PILR})& =  O\Big (\frac{1}{n_c}+\frac{1}{n_r} \Big), \\
  Var(\boldsymbol{\hat{\beta}}_{PILR})& =O\Big ( \frac{1}{n_c}+\frac{1}{n_r} \Big).
\end{align*}
Hence, $\hat{\mu}_{ILR}$ and $\boldsymbol{\hat{\beta}}_{ILR}$ are consistent estimates of, respectively, $\mu$ and $\boldsymbol{\beta}$ and $\hat{\mu}_{ILR}-\mu=O_p(\text{min}(n_c,n_r)^{-1/2})$ and $\boldsymbol{\hat{\beta}}_{ILR}-\boldsymbol{{\beta}}=O_p(\text{min}(n_c,n_r)^{-1/2})$.

\section{Simulation Results: Tables and plots } \label{sec:SimResultsTblPlots}

The Tables and Boxplots of this Section show the bulk of simulations results discussed in the main part of the paper.  
Additional results for scenarios S1-S4 are reported  in Tables \ref{tab:beta_frac_0.01} - \ref{tab:mu_frac_0.1}. 
Boxplots in Figures \ref{fig:est_beta0.01} -  \ref{fig:est_mu_0.1-01} show the distribution of relative biases of respective estimated parameters under  scenarios S1-S6. 

Under the high overlap and relatively large size of the reference sample (upper-left quadrant in each Table), estimates perform similarly well for all methods: estimates of both $\beta_{c1}$ and $\mu$ are nearly unbiased and variance estimates perform well providing for the nominal coverage by the normal confidence intervals. The only exception to this satisfactory performance are the estimates of ${\mu}$ in Table \ref{tab:mu_frac_mixed} that are positively biased by about $7\%$ for all the methods. Corresponding estimates for ${\beta}_{c1}$ (Table \ref{tab:beta_frac_mixed}) are essentially unbiased.

Results in the lower-left quadrants of the Tables and boxplots correspond to reference samples of smaller size $n_r=100$. Table \ref{tab:beta_frac_0.01} presents results for small sampling fractions $f_c=f_r=0.01$. It shows that ILR estimates of ${\beta}_{c1}$ are almost unbiased and about $30\%$ more efficient than the CLW based estimates. Also, the CLW based estimates are biased by about $7\%$. PILR estimates are of intermediate quality. Estimates corresponding to a larger sampling fractions $f_c=f_r=0.10$ in the lower-left quadrant of Table \ref{tab:beta_frac_0.1} show significantly smaller variation between different estimation methods. 
Numbers in the upper-right quadrants of these Tables 
correspond to relatively large reference samples and low overlap in covariate-defined domains. For small sampling fractions of $0.01$, ILR estimates of ${\beta}_{c1}$ have a significantly smaller bias and variance compared to the estimates from the CLW and PILR methods. This bias in the estimate of $\beta_{c1}$ also affects the variance estimates for the PILR and CLW methods, while variance estimates for the ILR perform satisfactory. For the larger sampling fraction of $0.10$, estimates based on the PILR and ILR methods are comparable and better than the CLW estimates. A similar pattern is observed in 
 the estimates of population mean ${\mu}$. 

Results in the lower-right quadrants of the Tables correspond to small reference samples sizes and low overlap in covariates domains between the non-probability and reference samples. For smaller sampling fractions of the non-probability sample,
estimates by the ILR method have significantly smaller bias and variance comparing to estimates by other methods. For a larger sampling fractions of $0.10$, PILR and ILR estimates in Tables \ref{tab:beta_frac_0.1} and \ref{tab:mu_frac_0.1} are comparable and significantly better than CLW estimates. Variances of ILR estimates are accurately estimated and produce close to nominal coverage of population parameters by the confidence intervals; 
PILR variances are reasonably accurate as well. However, for the CLW method, large biases of point estimates result in very inaccurate variance estimates.

\begin{table*}
\centering
\caption{ Estimate of propensity parameter ${\beta}_{c1}$ for sample fractions $f_c=f_r=0.01$ } 
\label{tab:beta_frac_0.01}
\begin{tabular}{c | c c  c c | c c  c c}
 
 & \multicolumn{4}{c|}{High overlap} & \multicolumn{4}{c}{Low overlap} \\
 & $Mean$ & $SE$ & $\overline{\widehat{SE}}$   &  $95\%CI$ & 
   $Mean$ & $SE$ & $\overline{\widehat{SE}}$   & $95\%CI$  \\  \hline 
  & \multicolumn{8}{c}{} \\ [-10 pt] 
 & \multicolumn{8}{c}{Scenario S1: $n_c=600, n_r=600$} \\
ILR & 1.01  & 0.07 & 0.07   & 0.95 & 1.01 & 0.08 & 0.08   & 0.96 \\
PILR & 1.01 & 0.08 & 0.08   & 0.95 & 1.05 & 0.18 & 0.14   & 0.88 \\
CLW & 1.01  & 0.09 & 0.09   & 0.94 & 1.07 & 0.21 & 0.16   & 0.86 \\
 & \multicolumn{8}{c}{Scenario S2: $n_c=100, n_r=100$} \\
ILR & 1.01  & 0.17 & 0.17   & 0.95 & 1.04 & 0.22 & 0.21   & 0.95 \\
PILR & 1.04 & 0.22 & 0.20   & 0.94 & 1.22 & 0.40 & 0.32   & 0.87 \\
CLW & 1.07  & 0.25 & 0.22   & 0.93 & 1.44 & 1.44 & 5.88   & 0.86 \\

\end{tabular}

\end{table*}


\begin{table*}
\centering
\caption{ Estimate of population mean  ${\mu}$ for sample fractions $f_c=f_r=0.01$ } 
\label{tab:mu_frac_0.01}

\begin{tabular}{c | c c  c c | c c  c c}
 
 & \multicolumn{4}{c|}{High overlap} & \multicolumn{4}{c}{Low overlap} \\
 & $Mean$ & $SE$ & $\overline{\widehat{SE}}$  &  $95\%CI$ & 
   $Mean$ & $SE$ & $\overline{\widehat{SE}}$  & $95\%CI$  \\  \hline 
  & \multicolumn{8}{c}{} \\ [-10 pt]
 & \multicolumn{8}{c}{Scenario S1: $n_c=600, n_r=600$} \\
ILR & 1.00  & 0.13 & 0.13   & 0.93 & 1.00 & 0.13 & 0.13   & 0.93 \\
PILR & 1.00 & 0.14 & 0.15   & 0.93 & 0.95 & 0.22 & 0.21   & 0.94 \\
CLW & 0.99  & 0.15 & 0.15   & 0.94 & 0.93 & 0.25 & 0.24   & 0.95 \\
 & \multicolumn{8}{c}{Scenario S2: $n_c=100, n_r=100$} \\
ILR & 1.05  & 0.32 & 0.30   & 0.89 & 1.03 & 0.32 & 0.31   & 0.91 \\
PILR & 1.01 & 0.38 & 0.42   & 0.91 & 0.87 & 0.52 & 1.03   & 0.93 \\
CLW & 0.99  & 0.41 & 0.53   & 0.91 & 0.77 & 0.64 & Inf    & 0.95 \\

\end{tabular}

\end{table*}


\begin{table*}
\centering
\caption{ Estimate of propensity parameter ${\beta}_{c1}$ for sample fractions $f_c=f_r=0.1$ }
\label{tab:beta_frac_0.1}
\begin{tabular}{c | c c  c c | c c  c c}
 
& \multicolumn{4}{c|}{High overlap} & \multicolumn{4}{c}{Low overlap} \\
 & $Mean$ & $SE$ & $\overline{\widehat{SE}}$   &  $95\%CI$ & 
   $Mean$ & $SE$ & $\overline{\widehat{SE}}$   & $95\%CI$  \\  \hline 
  & \multicolumn{8}{c}{} \\ [-10 pt]
 & \multicolumn{8}{c}{Scenario S3: $n_c=600, n_r=600$} \\
ILR & 1.00  & 0.08 & 0.08   & 0.94 & 1.00 & 0.10 & 0.09   & 0.94 \\
PILR & 1.00 & 0.08 & 0.08   & 0.95 & 1.01 & 0.13 & 0.12   & 0.93 \\
CLW & 1.01  & 0.09 & 0.10   & 0.95 & 1.05 & 0.20 & 0.17   & 0.91 \\
 & \multicolumn{8}{c}{Scenario S4: $n_c=100, n_r=100$} \\
ILR & 1.04  & 0.22 & 0.20   & 0.95 & 1.04 & 0.25 & 0.24   & 0.94 \\
PILR & 1.04 & 0.22 & 0.21   & 0.95 & 1.05 & 0.30 & 0.29   & 0.94 \\
CLW & 1.07  & 0.28 & 0.27   & 0.96 & 1.32 & 1.18 & 7.83   & 0.93 \\

\end{tabular}

\end{table*}


\begin{table*}
\centering
\caption{ Estimate of population mean  ${\mu}$ for sample fractions $f_c=f_r=0.1$ } 
\label{tab:mu_frac_0.1}

\begin{tabular}{c | c c  c c | c c  c c}
 
 & \multicolumn{4}{c|}{High overlap} & \multicolumn{4}{c}{Low overlap} \\
 & $Mean$ & $SE$ & $\overline{\widehat{SE}}$   &  $95\%CI$ & 
   $Mean$ & $SE$ & $\overline{\widehat{SE}}$   & $95\%CI$  \\  \hline 
  & \multicolumn{8}{c}{} \\ [-10 pt]
 & \multicolumn{8}{c}{Scenario S3: $n_c=600, n_r=600$} \\
ILR & 1.02  & 0.12 & 0.11   & 0.91 & 0.97 & 0.12 & 0.11   & 0.93 \\
PILR & 1.01 & 0.12 & 0.12   & 0.91 & 0.96 & 0.15 & 0.14   & 0.95 \\
CLW & 1.01  & 0.14 & 0.14   & 0.91 & 0.92 & 0.21 & 0.21   & 0.95 \\
 & \multicolumn{8}{c}{Scenario S4: $n_c=100, n_r=100$} \\
ILR & 1.03  & 0.27 & 0.26   & 0.90 & 1.00 & 0.29 & 0.25   & 0.89 \\
PILR & 1.02 & 0.29 & 0.31   & 0.90 & 0.97 & 0.37 & 0.37   & 0.91 \\
CLW & 1.00  & 0.33 & 0.41   & 0.91 & 0.81 & 0.61 & Inf    & 0.93 \\

\end{tabular}

\end{table*}


\begin{figure}
    \centering
    \includegraphics[width=1\linewidth]{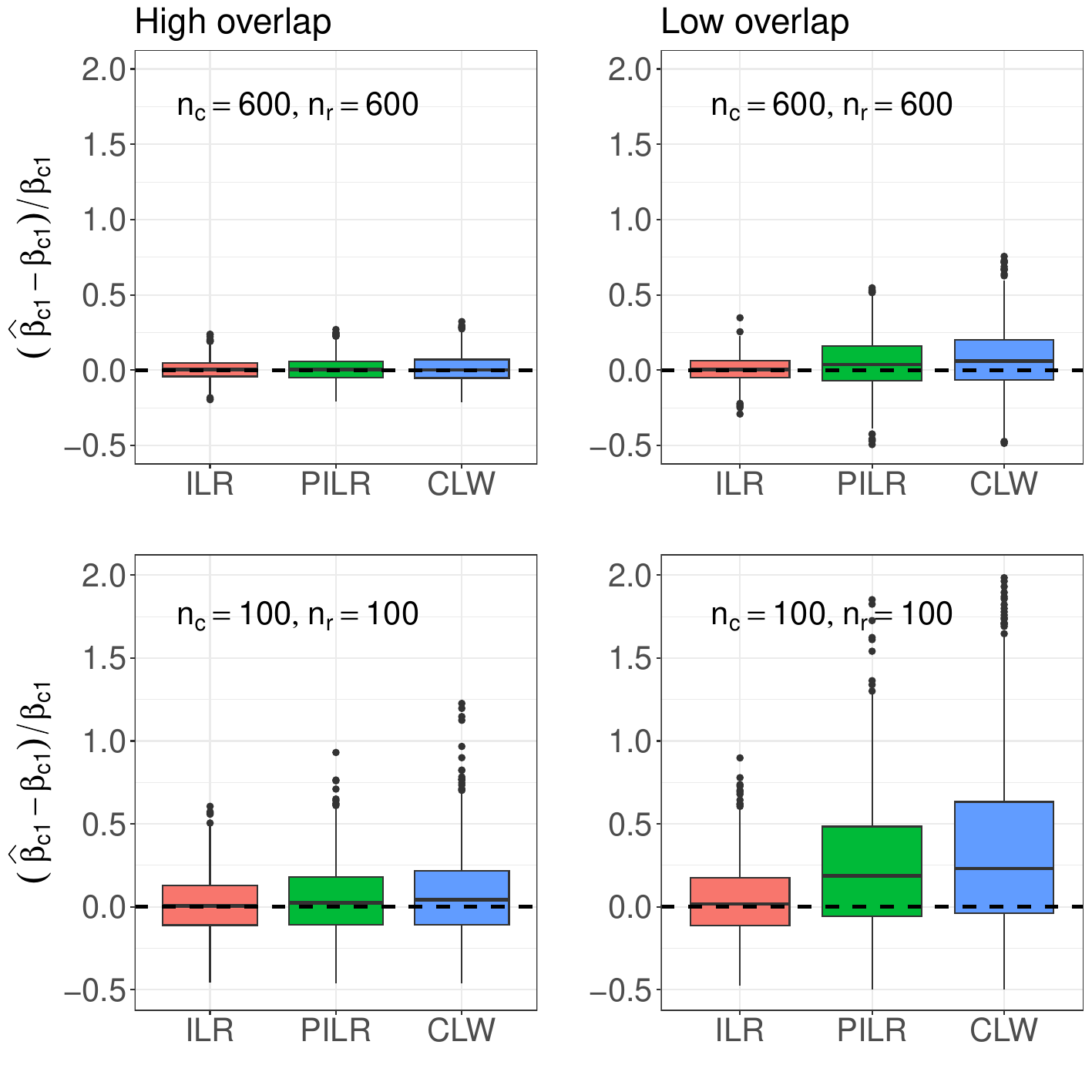}
    \caption{Relative bias of the estimated propensity parameters ${\beta}_{c1}$ for sample fraction 0.01 over the MC simulations, scenarios S1 and S2}
    \label{fig:est_beta0.01}
\end{figure}


\begin{figure}
    \centering
    \includegraphics[width=1\linewidth]{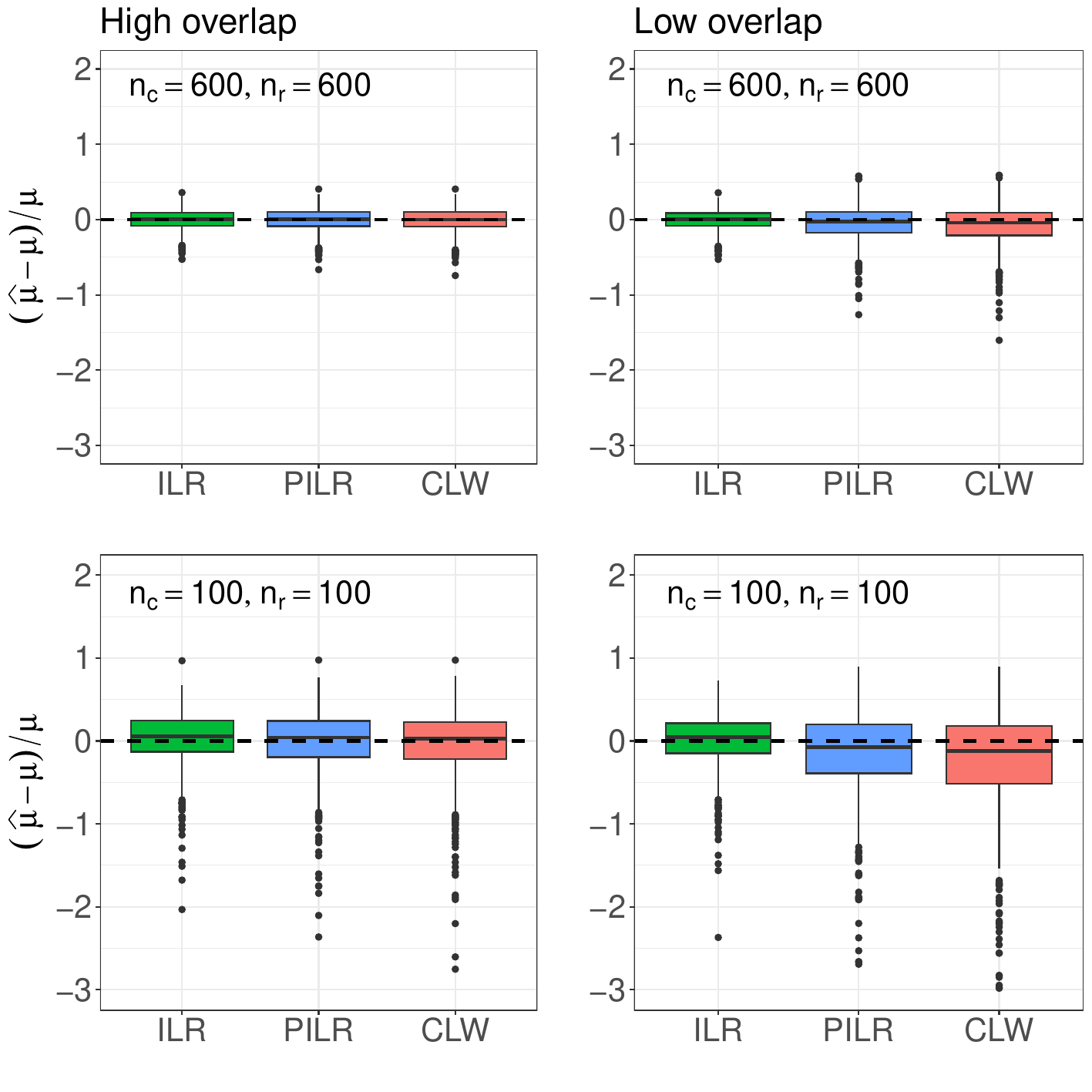}
    \caption{Relative bias of the estimated population mean  ${\mu}$ for sample fraction 0.01 over the MC simulations, scenarios S1 and S2}
    \label{fig:est_mu_0.01}
\end{figure}


\begin{figure}
    \centering
    \includegraphics[width=1\linewidth]{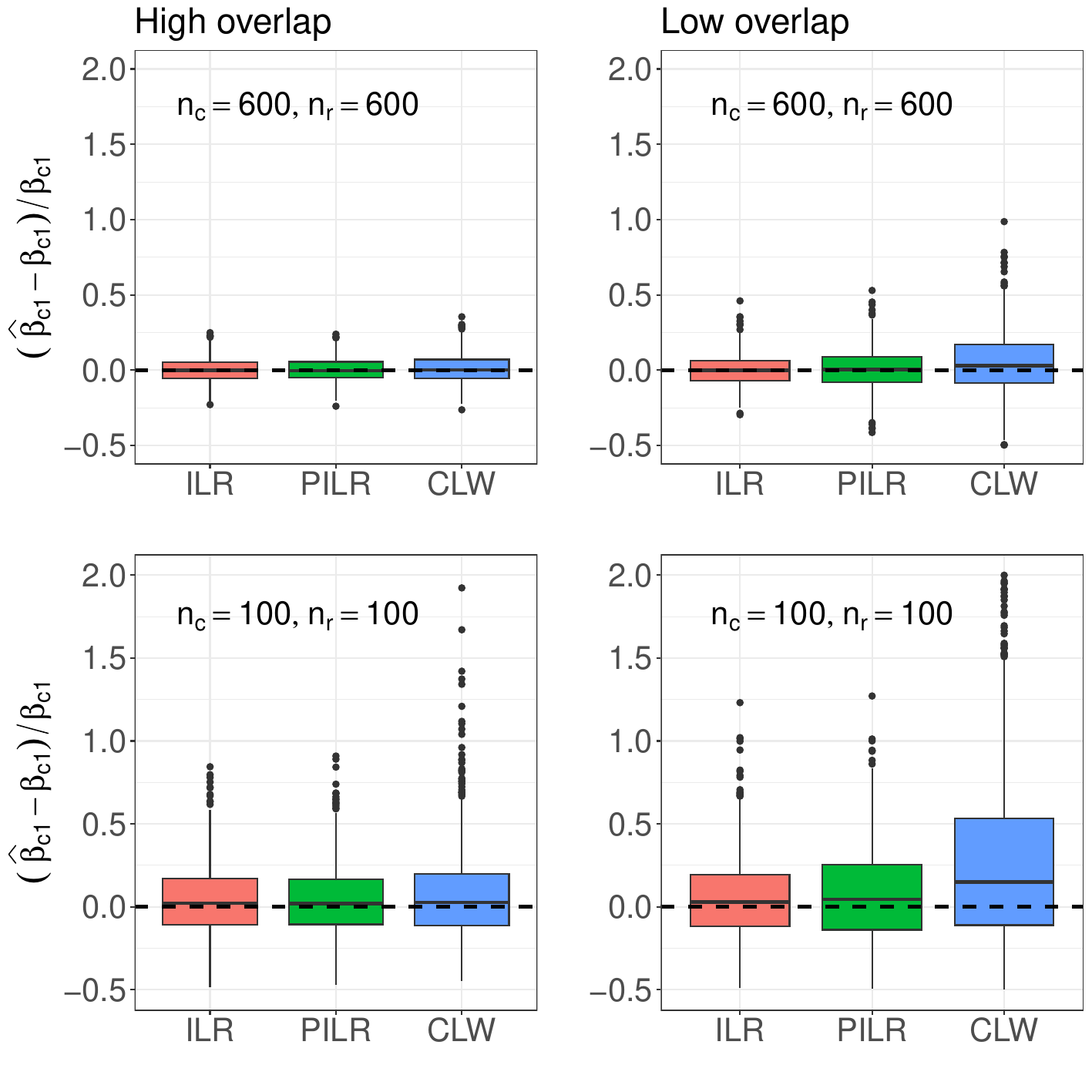}
    \caption{Relative bias of the estimated propensity parameters ${\beta}_{c1}$  for sample fraction 0.1 over the MC simulations, scenarios S3 and S4}
    \label{fig:est_beta_0.1}
\end{figure}


\begin{figure}
    \centering
    \includegraphics[width=1\linewidth]{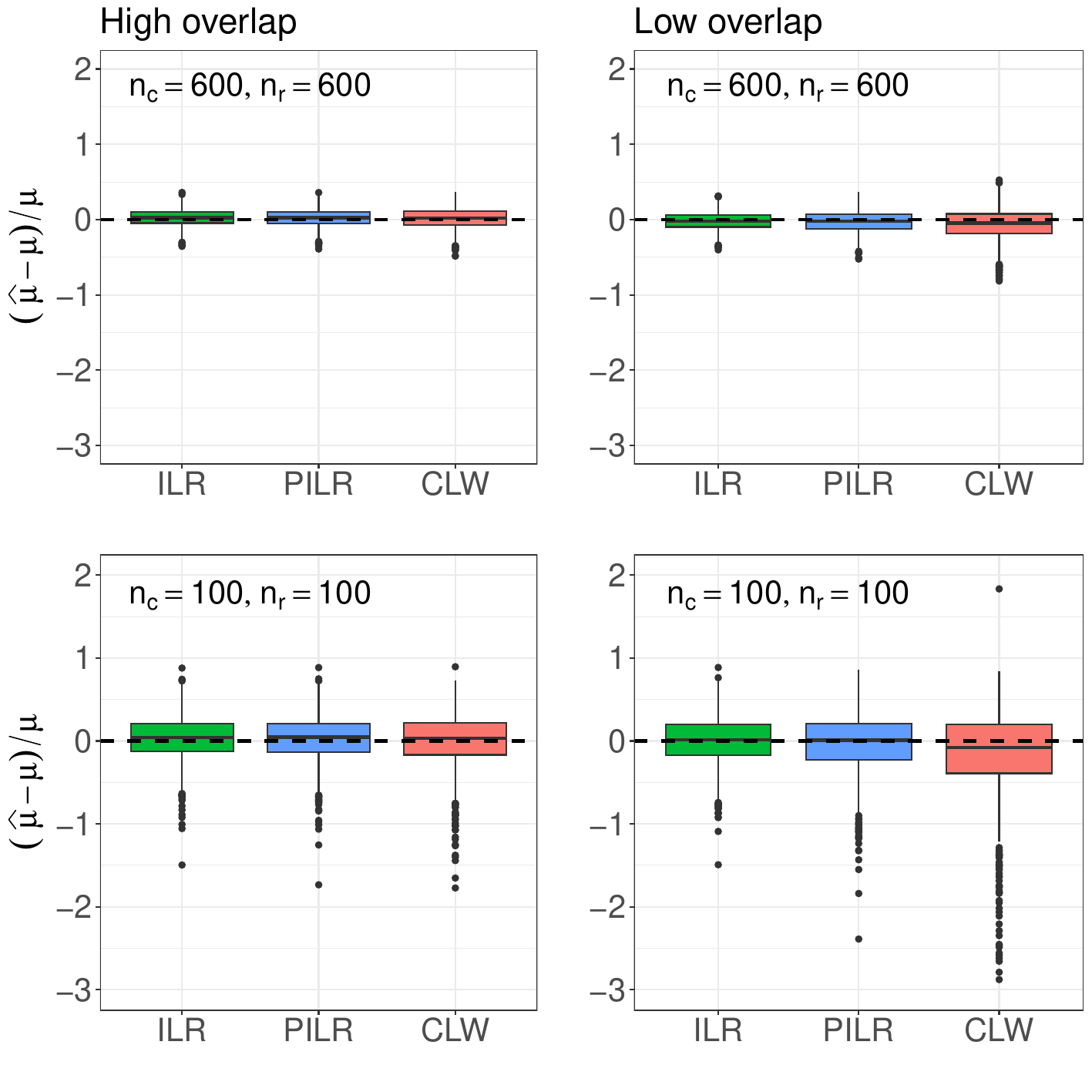}
    \caption{Relative bias of the estimated population mean  ${\mu}$ for sample fraction 0.1 over the MC simulations, scenarios S3 and S4}
    \label{fig:est_mu_0.1}
\end{figure}


\begin{figure}
    \centering
    \includegraphics[width=1\linewidth]{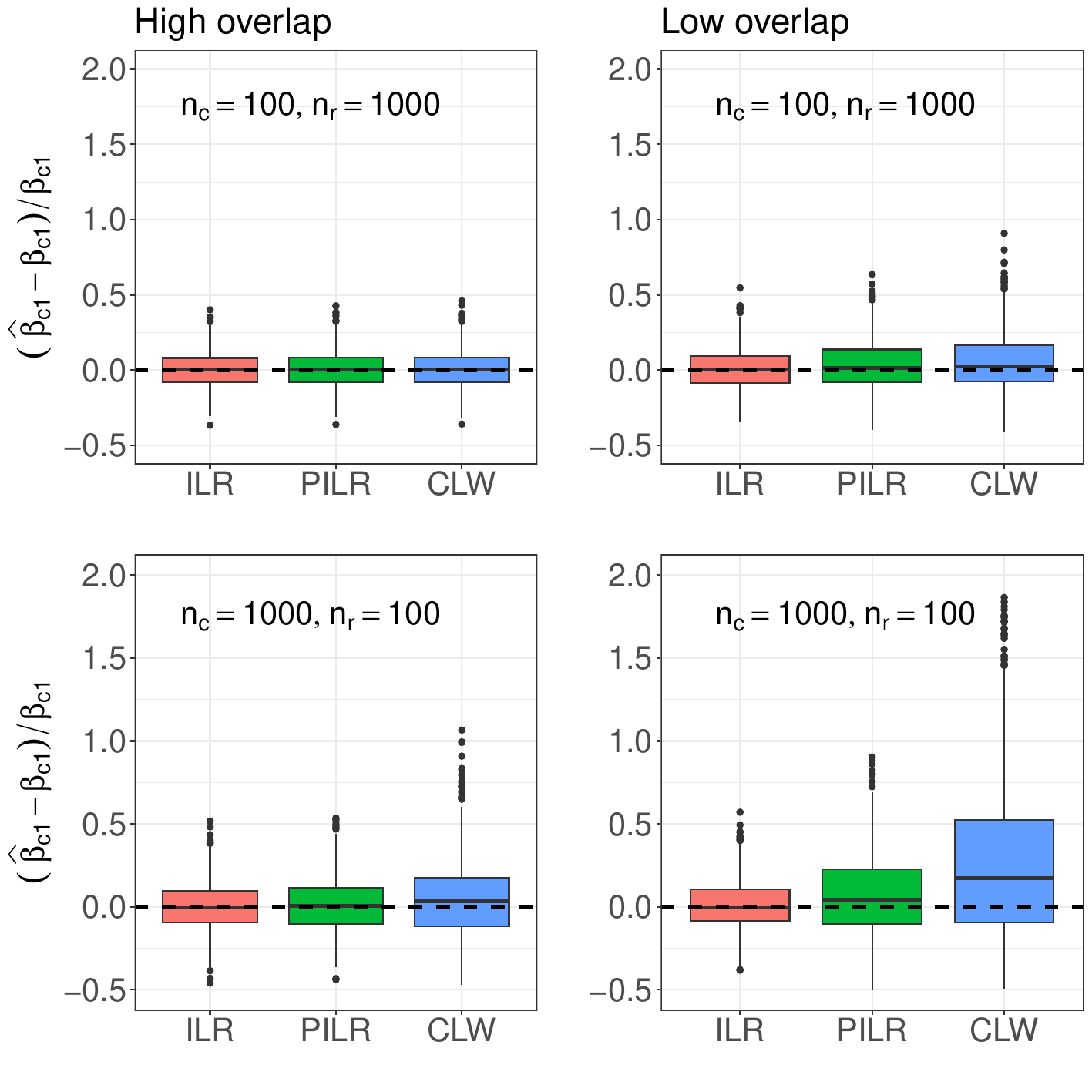}
    \caption{Relative bias of the estimated propensity parameters ${\beta}_{c1}$  for sampling fractions $\left(f_c,f_r\right)$  equal to $\left(0.01, 0.1 \right)$ (upper row) and to $\left(0.1, 0.01 \right)$ (lower row) over the MC simulations, scenarios S5 and S6}
    \label{fig:est_beta_0.1-01}
\end{figure}


\begin{figure}
    \centering
    \includegraphics[width=1\linewidth]{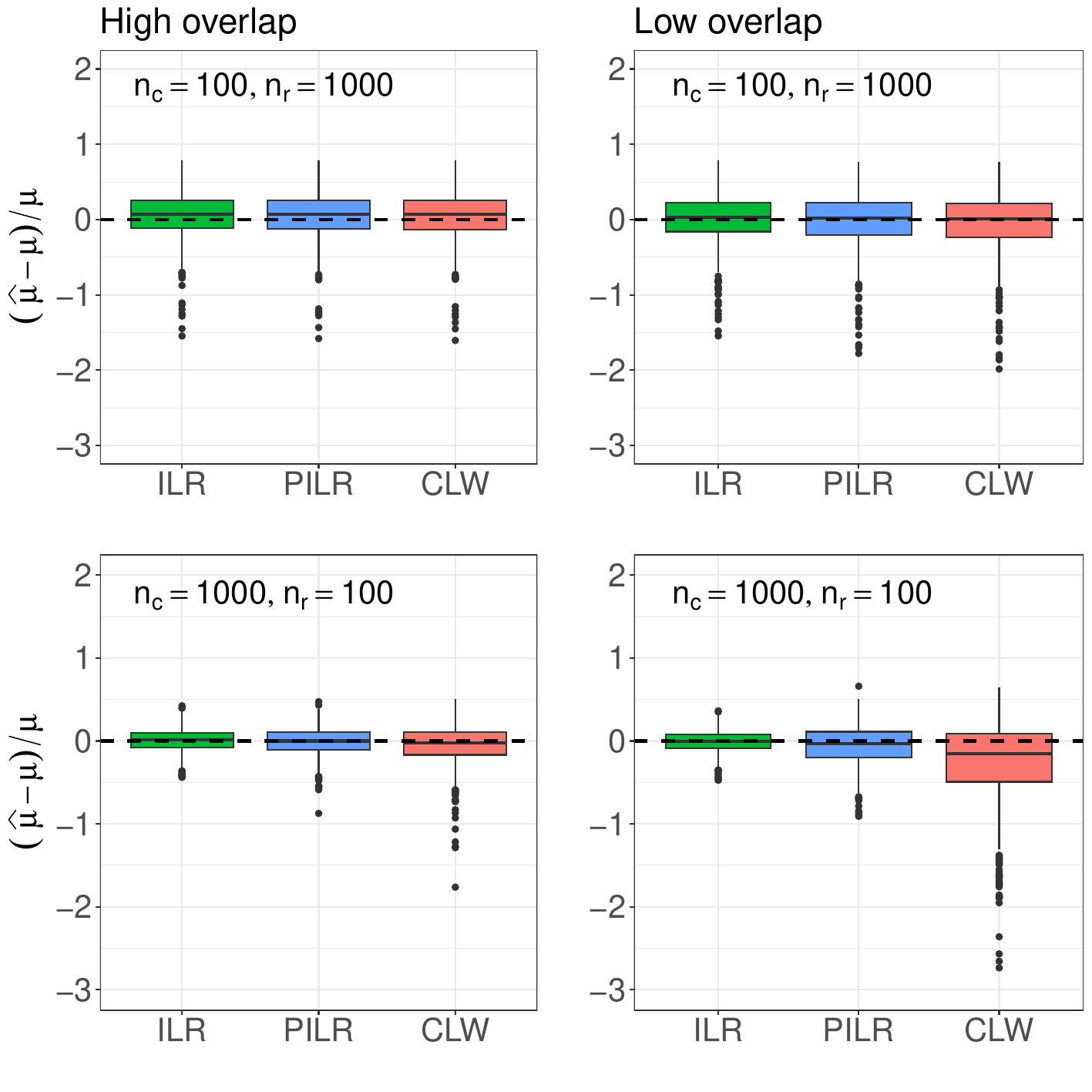}
    \caption{Relative bias of the estimated population mean  ${\mu}$ for sampling fractions $\left(f_c,f_r\right)$  equal to $\left(0.01, 0.1 \right)$ (upper row) and to $\left(0.1, 0.01 \right)$ (lower row) over the MC simulations, scenarios S5 and S6}
    \label{fig:est_mu_0.1-01}
\end{figure}


\section{On a bias in the two-step estimation of response propensity} \label{sec:step1_vs_step2}
In Section \ref{sec:PILR}, the PILR approach is presented as a one-step alternative to the ALP weighting method of \citet{2021valliant}. The methods use identical pseudo-likelihoods  \eqref{eq:wang2} and \eqref{eq:wang3} and similar relationship \eqref{eq:identityWVL} between $\pi_{\delta}$ and participation probability $\pi_c$. The methods differ in that ALP estimates participation probability $\pi_c$ in two steps while  PILR estimates $\pi_c$ in one step. 

In Figure \ref{fig:step1vs2}, we show predicted vs true $\pi_c$ for the entire population using three one-step estimation methods presented in this paper, namely CLW, ILR and PILR, and the two-step ALP method by \citet{2021valliant} for simulation scenarios S5 and S6 described in Table \ref{tab:sim_scenarios} and for low and high overlaps in covariate domains. 

\begin{figure}
    \centering
    \includegraphics[width=1\linewidth]{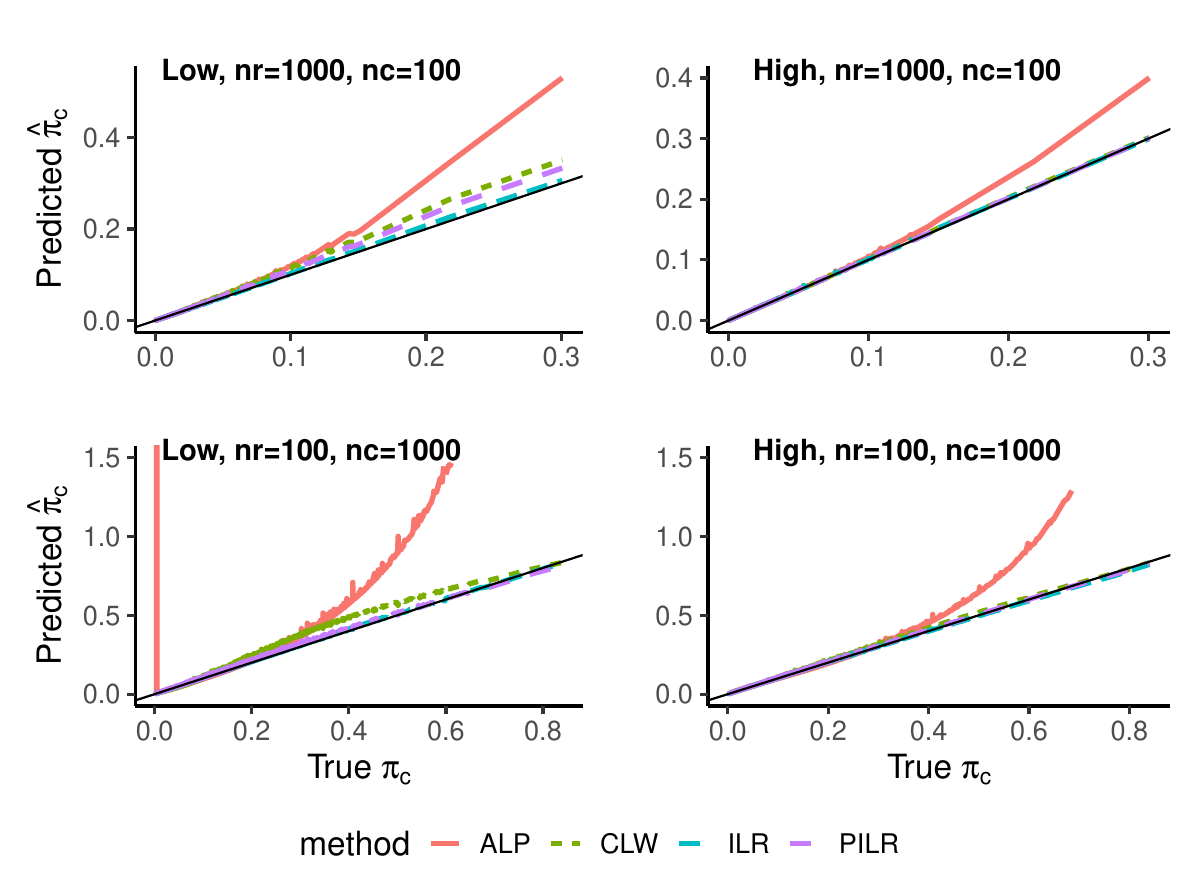}
    \caption{Predicted vs true participation probability $\pi_c$ obtained using one- and two-step estimation methods. Scenarios S5 and S6 of Table \ref{tab:sim_scenarios} for Low and High overlaps in covariate domains}
    \label{fig:step1vs2}
\end{figure}

All methods are approximately unbiased for small values of true $\pi_c$, while the two-step ALP method systematically overestimates participation probabilities of population units with large true $\pi_c$. 

At the same time, ALP produced estimates of the finite population mean \eqref{eq:hajek} similar to the PILR estimates shown in Table \ref{tab:mu_frac_mixed}. This could be a feature of this particular simulation setup. In general, biased estimates of participation probabilities may result in biased estimates of the population mean from non-probability samples.     

\end{appendix}

\end{document}